\documentclass[copyright,creativecommons]{eptcs}
\providecommand{\publicationstatus}{An extended abstract (all but the appendices) will appear in\\
  G. Caltais and C. A. Mezzina (Eds): Combined Workshop on\\
Expressiveness in Concurrency and Structural Operational Semantics\\(EXPRESS/SOS 2023).}
\hypersetup{hidelinks}

\usepackage{mathpartir}
\usepackage{listings}

\usepackage{fge}
\usepackage{amsmath,amssymb,latexsym}
\usepackage{mathtools}

\usepackage{enumitem}
\usepackage{xcolor}
\definecolor{Green}{cmyk}{1,0,0.5,0.3}
\definecolor{orange}{cmyk}{0,0.82,1,0.01}
\definecolor{Purple}{cmyk}{0.65,1,0,0}
\makeatletter
\def\@listi{\leftmargin\leftmargini
            \parsep 0\p@ \@plus1.5\p@ \@minus0\p@
            \topsep 2\p@   \@plus2\p@ \@minus2\p@
            \itemsep0\p@ \@plus1.5\p@ \@minus0\p@}
\let\@listI\@listi
\@listi
\makeatother
\makeatletter
\newif\if@qeded
\global\@qededtrue
\def\qed{\hfill$\Box$\global\@qededtrue}
\def\qedneeded{\global\@qededfalse}
\def\qedifneeded{\if@qeded\else\qed\fi}
\makeatother
\newtheorem{theo}{Theorem}[section]
\newtheorem{conj}[theo]{Conjecture}
\newtheorem{defi}[theo]{Definition}
\newtheorem{lemm}[theo]{Lemma}
\newtheorem{exam}[theo]{Example}
\newtheorem{obse}[theo]{Observation}

\newenvironment{definitionc}[1]{\begin{defi}[#1] \rm }{\end{defi}}
\newenvironment{conjecture}    {\begin{conj}     \rm }{\end{conj}}

\newenvironment{theoremc}[1]   {\begin{theo}[#1] \rm }{\end{theo}}
\newenvironment{lemma}         {\begin{lemm}     \rm }{\end{lemm}}
\newenvironment{example}       {\begin{exam}     \rm }{\end{exam}}
\newenvironment{examplec}[1]   {\begin{exam}[#1] \rm }{\end{exam}}
\newenvironment{observation}   {\begin{obse}     \rm }{\end{obse}}
\newenvironment{proof}{\qedneeded\begin{trivlist} \item[\hspace{\labelsep}\bf Proof:]}
               {\qedifneeded\end{trivlist}}

\newcommand{\df}[1]{Definition~\ref{df:#1}}
\newcommand{\dfs}[1]{Definitions~#1}
\newcommand{\lem}[1]{Lemma~\ref{lem:#1}}

\newcommand{\thm}[1]{Theorem~\ref{thm:#1}}

\newcommand{\obs}[1]{Observation~\ref{obs:#1}}

\newcommand{\ex}[1]{Example~\ref{ex:#1}}

\newcommand{\Sec}[1]{Section~\ref{sec:#1}}

\newcommand{\app}[1]{Appendix~\ref{app:#1}}

\newcommand{\fig}[1]{Figure~\ref{fig:#1}}

\newcommand{\tab}[1]{Table~\ref{tab:#1}}

\usepackage{xspace}
\usepackage{wrapfig}
\newcommand{\its}{\ensuremath{(\Sigma,\Sigma_\TT)}\xspace}
\newcommand{\itt}{\ensuremath{(\Sigma,\RR)}\xspace}
\newcommand{\ite}{\ensuremath{(\Sigma_\PP,\Sigma_\TT)}\xspace}
\newcommand{\tss}{\ensuremath{(\Sigma,\RR,\scn)}\xspace}
\newcommand{\tsss}{\ensuremath{(\Sigma,\RR,\scn,\UU)}\xspace}
\newcommand{\lit}{\ensuremath{\Sigma}\xspace}

\newcommand{\plat}[1]{\raisebox{0pt}[0pt][0pt]{#1}} 

\DeclareSymbolFont{frenchscript}{OMS}{ztmcm}{m}{n}    
\DeclareMathSymbol{\Lab}{\mathord}{frenchscript}{76}  
\DeclareMathSymbol{\Pow}{\mathord}{frenchscript}{80}  
\DeclareMathSymbol{\R}{\mathord}{frenchscript}{82}    
\DeclareMathAlphabet{\altmathcal}{OMS}{cmsy}{m}{n}    
\DeclareMathAlphabet{\mathbbm}{U}{bbm}{m}{n}          
\newcommand{\IN}{\mathbbm{N}}                         

\newcommand{\djcup}{\mathbin{\mathaccent\cdot\cup}} 

\newcommand{\obis}[2]{\mathrel{_{#1}\,                              
  \raisebox{.3ex}{$\underline{\makebox[.7em]{$\leftrightarrow$}}$}
  \,_{#2}}}
\newcommand{\bis}[1]{\obis{}{#1}}                                   
\newcommand{\bisep}{\bis{\mathit{ep}}}                              

\newcommand{\Tr}{\mathit{Tr}}                         
\newcommand{\source}{\mathit{source}}                 
\newcommand{\target}{\mathit{target}}                 
\newcommand{\en}{\mathit{en}}                         
\newcommand{\goto}[1]{\stackrel{#1}{\longrightarrow}} 

\newcommand{\PP}{\altmathcal{P}}                
\newcommand{\Var}{\altmathcal{V}_\PP}           
\newcommand{\var}{\mathit{var}}                 
\newcommand{\Op}{\mathit{Op}}                   
\newcommand{\rec}[1]{\mathbbm\langle #1\rangle} 
\newcommand{\IP}{\mathbbm{P}}                   
\newcommand{\cP}{\mathrm{P}}                    

\DeclareMathSymbol{\B}{\mathord}{frenchscript}{66}            
\DeclareMathSymbol{\Ch}{\mathord}{frenchscript}{67}           
\DeclareMathSymbol{\Sig}{\mathord}{frenchscript}{83}          
\newcommand{\ABCdE}{ABCdE\xspace}                             
\newcommand{\signals}{\mathrel{\hat{}\!}}                     
\newcommand{\actsyn}[1]{\mathord{\stackrel{#1}{\rightarrow}}} 
\newcommand{\sigsyn}[1]{^{\rightarrow #1}}                    
\newcommand{\Left}{\mathrm{\scriptscriptstyle L}}             
\newcommand{\Right}{\mathrm{\scriptscriptstyle R}}            

\newcommand{\sR}{O}               
\newcommand{\RR}{\altmathcal{R}}  
\newcommand{\UU}{\altmathcal{U}}  
\newcommand{\scr}{r}              
\newcommand{\snr}{\textsc{\textcolor{blue}{r}}}     
\newcommand{\sns}{\textsc{\textcolor{blue}{s}}}     
\newcommand{\scn}{\textsc{n}}     
\newcommand{\NN}{\altmathcal{N}}  

\newcommand{\Act}{\mathit{Act}}   
\newcommand{\Actc}{\mathit{Act}_{CCS}}
\newcommand{\In}{\mathit{In}}     

\newcommand{\aconc}{\mathrel{\mbox{$\smile\hspace{-1.25ex}\raisebox{3pt}{$\scriptscriptstyle\bullet$}$}}}
\newcommand{\naconc}{\mathrel{\mbox{$\,\not\hspace{-1pt}\smile\hspace{-1.25ex}\raisebox{3pt}{$\scriptscriptstyle\bullet$}$}}}     

\newcommand{\VV}{\altmathcal{V}}                            
\newcommand{\px}{\mathit{px}}                               

\newcommand{\TT}{\altmathcal{T}}                            
\newcommand{\TVar}{\altmathcal{V}_\TT}                      
\newcommand{\Tvar}{\mathit{var}_{\scriptscriptstyle\!\TT}}  
\newcommand{\tx}{\mathit{tx}}                               
\newcommand{\ty}{\mathit{ty}}                               
\newcommand{\tz}{\mathit{tz}}                               
\newcommand{\ux}{\mathit{ux}}                               
\newcommand{\vx}{\mathit{vx}}                               
\newcommand{\IT}{\mathbbm{T}}                               
\newcommand{\cT}{\mathrm{T}}                                
\newcommand{\Osrc}{\mathit{src}_\circ}                      
\newcommand{\Otar}{\mathit{tar}_\circ}                      
\newcommand{\Oell}{\ell_\circ}                              
\newcommand{\Oen}{\en_\circ}                                
\newcommand{\Vtar}[1][\mathring t]{T_{\scriptstyle #1}}     

\newcommand{\ITE}{\mathbbm{T\!E}}                           
\newcommand{\RecAct}{\mathit{rec}_\Act}                     
\newcommand{\RecIn}{\mathit{rec}_\In}                       

\newcommand{\Tsigma}{\sigma_{\scriptscriptstyle\!\TT}}      
\newcommand{\Trho}{\rho_{\scriptscriptstyle\!\TT}}          
\newcommand{\Trhoi}{\rho_{{\scriptscriptstyle\!\TT}i}}      
\newcommand{\Trhoj}{\rho_{{\scriptscriptstyle\!\TT}j}}      
\newcommand{\Tnu}{\nu_{\scriptscriptstyle\!\TT}}            
\newcommand{\Tnui}{\nu_{{\scriptscriptstyle\!\TT}i}}        
\newcommand{\Tnuj}{\nu_{{\scriptscriptstyle\!\TT}j}}        

\newcommand{\mylabelA}[2]{
  \expandafter\ifx\csname #1\endcsname\relax
    \expandafter\gdef\csname #1\endcsname{#2} \hypertarget{lab:#1}{~{\color{blue} \csname #1\endcsname}}
  \else
    \ClassError{Weiyou}{Label #1 has already been used}{}
  \fi}
\newcommand{\myrefA}[1]{\hyperlink{lab:#1}{{\color{blue} \csname #1\endcsname}}}

\newcommand{\mylabelS}[1]{\hypertarget{lab:#1}{~{\color{orange} #1}}}
\newcommand{\myrefS}[1]{\hyperlink{lab:#1}{{\color{orange} #1}}}
\newcommand{\myrefD}[2]{\hyperlink{lab:#1}{{\color{orange} #2}}}

\def\titlerunning{A Lean-Congruence Format for EP-Bisimilarity}
\title\titlerunning
\def\authorrunning{R.J. van Glabbeek, P. H\"ofner \& W. Wang}
\author{Rob van Glabbeek%
  \thanks{Supported by Royal Society Wolfson Fellowship RSWF{\textbackslash}R1{\textbackslash}221008}\,
  \href{https://orcid.org/0000-0003-4712-7423}{\includegraphics[scale=.04]{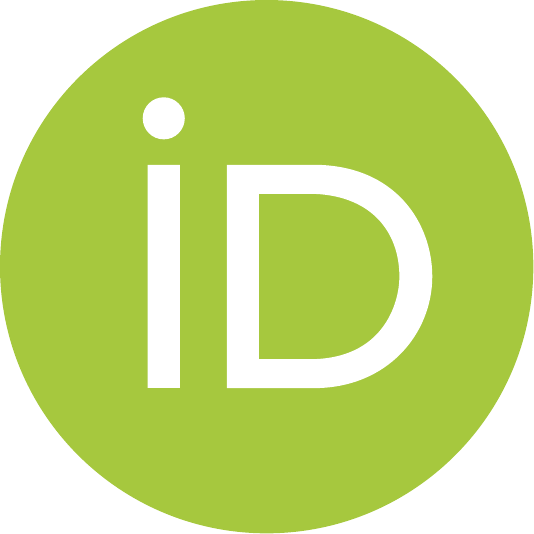}}
  \institute{School of Informatics\\ University of Edinburgh, UK}
  \institute{School of Computer Science and Engineering\\ University of New South Wales\\ Sydney, Australia}
  \email{rvg@cs.stanford.edu}
\and Peter H\"ofner\,%
  \href{https://orcid.org/0000-0002-2141-5868}{\includegraphics[scale=.04]{orcid.pdf}}
  \qquad\qquad Weiyou Wang
  \institute{School of Computing\\ Australian National University\\ Canberra, Australia}
  \email{peter.hoefner@anu.edu.au}\email{weiyou.wang@anu.edu.au}
}

\begin{document}
\maketitle 

\begin{abstract}
Enabling preserving bisimilarity is a refinement of strong bisimilarity, which preserves safety as well as liveness properties. 
To define it properly, labelled transition systems needed to be upgraded with a successor relation,
capturing concurrency between transitions enabled in the same state. We enrich the well-known De
Simone format to handle inductive definitions of this successor relation.
We then establish that ep-bisimilarity is a congruence for the operators, as well as lean congruence for recursion, for all (enriched) De Simone languages.
\end{abstract}

\section{Introduction\label{sec:intro}}
Recently, we introduced a finer alternative to strong bisimilarity, called enabling preserving bisimilarity. 
The motivation behind this concept was to preserve liveness properties, which are \emph{not} always preserved
by classical semantic equivalences, including strong bisimilarity.

\begin{examplec}{\cite{GHW21ea}}\label{ex:filippo}
Consider the following two programs, and assume that all variables are initialised to \lstinline{0}.
\vspace{-7.5mm}
\begin{center}\hfill\hfill\hfill
	\begin{minipage}[t]{0.3\textwidth}
\begin{lstlisting}[basicstyle=\fontsize{8}{9}\selectfont]
while(true) do
  choose
    if true then y := y+1;
    if x = 0 then x := 1;
  end
od
\end{lstlisting}
	\end{minipage}\hfill\hfill
        \begin{minipage}[t]{0.15\textwidth}
          \expandafter\ifx\csname graph\endcsname\relax
   \csname newbox\expandafter\endcsname\csname graph\endcsname
\fi
\ifx\graphtemp\undefined
  \csname newdimen\endcsname\graphtemp
\fi
\expandafter\setbox\csname graph\endcsname
 =\vtop{\vskip 0pt\hbox{%
\pdfliteral{
q [] 0 d 1 J 1 j
0.576 w
0.576 w
24.048 -4.824 m
24.048 -7.488222 21.888222 -9.648 19.224 -9.648 c
16.559778 -9.648 14.4 -7.488222 14.4 -4.824 c
14.4 -2.159778 16.559778 0 19.224 0 c
21.888222 0 24.048 -2.159778 24.048 -4.824 c
S
0.072 w
q 0 g
7.2 -3.024 m
14.4 -4.824 l
7.2 -6.624 l
7.2 -3.024 l
B Q
0.576 w
0 -4.824 m
7.2 -4.824 l
S
0.072 w
q 0 g
26.424 -14.544 m
22.608 -8.208 l
28.944 -12.024 l
26.424 -14.544 l
B Q
0.576 w
15.84 -8.208 m
10.332 -13.716 l
6.58656 -17.46144 4.824 -21.528 4.824 -26.424 c
4.824 -31.32 9.432 -33.624 19.224 -33.624 c
29.016 -33.624 33.624 -31.32 33.624 -26.424 c
33.624 -21.528 31.94208 -17.54208 28.368 -13.968 c
23.112 -8.712 l
S
Q
}%
    \graphtemp=.5ex
    \advance\graphtemp by 0.547in
    \rlap{\kern 0.267in\lower\graphtemp\hbox to 0pt{\hss \footnotesize $y:=y+1$\hss}}%
\pdfliteral{
q [] 0 d 1 J 1 j
0.576 w
72 -4.824 m
72 -7.488222 69.840222 -9.648 67.176 -9.648 c
64.511778 -9.648 62.352 -7.488222 62.352 -4.824 c
62.352 -2.159778 64.511778 0 67.176 0 c
69.840222 0 72 -2.159778 72 -4.824 c
S
0.072 w
q 0 g
55.224 -3.024 m
62.424 -4.824 l
55.224 -6.624 l
55.224 -3.024 l
B Q
0.576 w
23.976 -4.824 m
55.224 -4.824 l
S
Q
}%
    \graphtemp=\baselineskip
    \multiply\graphtemp by -1
    \divide\graphtemp by 2
    \advance\graphtemp by .5ex
    \advance\graphtemp by 0.067in
    \rlap{\kern 0.600in\lower\graphtemp\hbox to 0pt{\hss \footnotesize $x:=1$\hss}}%
\pdfliteral{
q [] 0 d 1 J 1 j
0.576 w
0.072 w
q 0 g
74.448 -14.544 m
70.56 -8.208 l
76.968 -12.024 l
74.448 -14.544 l
B Q
0.576 w
63.792 -8.208 m
58.284 -13.716 l
54.53856 -17.46144 52.776 -21.528 52.776 -26.424 c
52.776 -31.32 57.384 -33.624 67.176 -33.624 c
76.968 -33.624 81.576 -31.32 81.576 -26.424 c
81.576 -21.528 79.9056 -17.54208 76.356 -13.968 c
71.136 -8.712 l
S
Q
}%
    \graphtemp=.5ex
    \advance\graphtemp by 0.547in
    \rlap{\kern 0.933in\lower\graphtemp\hbox to 0pt{\hss \footnotesize $y:=y+1$\hss}}%
    \hbox{\vrule depth0.547in width0pt height 0pt}%
    \kern 1.114in
  }%
}%

          \vspace{1ex}
          \box\graph
        \end{minipage}\hfill\hfill\hfill
	\begin{minipage}[t]{0.15\textwidth}
\footnotesize\begin{lstlisting}[basicstyle=\fontsize{8}{9}\selectfont]
while(true) do
  y := y+1;
od
\end{lstlisting}
	\end{minipage}
	\raisebox{-27pt}{\scalebox{1.3}[2.8]{$\|$}}
	\begin{minipage}[t]{0.1\textwidth} 
\footnotesize\begin{lstlisting}[basicstyle=\fontsize{8}{9}\selectfont]
x := 1;
\end{lstlisting}
	\end{minipage}
	\hfill\mbox{}
\end{center}
\vspace{-7.5mm}
\noindent The code on the left-hand side presents a non-terminating while-loop offering an internal nondeterministic choice.
The conditional \lstinline{if x = 0 then x := 1} describes an atomic read-modify-write operation.\footnote{\url{https://en.wikipedia.org/wiki/Read-modify-write}}
Since the non-deterministic choice does not guarantee to ever pick the second conditional, this example should not satisfy the liveness property `eventually \lstinline{x=1}'.

The example on the right-hand side is similar, but here two different components handle the variables \lstinline{x} and 
\lstinline{y} separately. The two programs should be considered independent -- by default we assume they are executed on different cores. Hence the property `eventually \lstinline{x=1}' should hold.

The two programs behave differently with regards to (some) liveness properties. However, it is easy to
verify that they are strongly bisimilar, when considering the traditional modelling of such code in
terms of transition systems. In fact, their associated transition systems, also displayed above, are identical. 
Hence, strong bisimilarity does not preserve all liveness properties.
\end{examplec}

\noindent
Enabling preserving bisimilarity (ep-bisimilarity) --~see next section for a formal definition~-- distinguishes these examples 
and preserves liveness.
In contrast to classical  bisimulations, which are relations of type $\text{States}\times \text{States}$,
this equivalence is based on triples. 
An ep-bisimulation additionally maintains for each pair of related states $p$ and $q$ a relation $R$ between the
transitions enabled in $p$ and $q$, and this relation should be preserved when matching related
  transitions in the bisimulation game. When formalising this, we need transition systems upgraded
  with a \emph{successor relation} that matches each transition $t$ enabled in a state $p$ to a
  transition $t'$ enabled in $p'$, when performing a transition from $p$ to $p'$ that does not affect
  $t$. Intuitively, $t'$ describes the same system behaviour as $t$, but the two transitions could be
  formally different as they may have different sources.
It is this successor relation that distinguishes the transition systems in the example~above.

In~\cite{GHW21ea}, we showed that ep-bisimilarity is a congruence for all operators of Milner's
Calculus of Communication Systems (CCS), enriched with a successor relation. We extended this result to the Algebra of Broadcast Communication with discards and Emissions (\ABCdE), an extension of CCS with broadcast communication, discard actions and signal emission. 
\ABCdE\ subsumes many standard process algebras found in the literature.

In this paper, we introduce a new congruence format for structural operational semantics, which is based on the well-known De Simone Format and respects the successor relation. This format allows us to generalise the results of \cite{GHW21ea} in two ways:
first, we prove that ep-bisimilarity is a congruence for all operators of \emph{any} process
algebras that can be formalised in the De Simone format with successors. Applicable languages include CCS and \ABCdE. 
Second, we show that ep-bisimilarity is a lean congruence for recursion~\cite{vG17b}. 
Here, a lean congruence preserves equivalence when replacing closed subexpressions of a process by equivalent alternatives.

\section{Enabling Preserving Bisimilarity\label{sec:ep}}

To build our abstract theory of De Simone languages and De Simone formats, we briefly recapitulate the definitions of labelled transition systems with successors, and ep-bisimulation. A detailed description can be found in \cite{GHW21ea}.

  A \emph{labelled transition system (LTS)}\index{labelled transition system} is a tuple
 $(S,\Tr,\source,\target,\ell)$ with
    $S$ and $\Tr$ sets of \emph{states}\index{state} and \emph{transitions}\index{transition},
    $\source,\target:\Tr\to S$\index{source}\index{target} and $\ell:\Tr\to\Lab$\index{transition label},
    for some set $\Lab$ of transition labels.
  A transition $t\in\Tr$ of an LTS is \emph{enabled}\index{enabled} in a state $p\in S$ if $\source(t)=p$.
  The set of transitions enabled in $p$ is~$\en(p)$.

\begin{definitionc}{LTSS \cite{GHW21ea}}\label{df:LTSS}\upshape
  A \emph{labelled transition system with successors\index{LTSS} (LTSS)}
    is a tuple $(S,\Tr,\source,\linebreak[1]\target,\ell,\leadsto)$ with
    $(S,\Tr,\source,\target,\ell)$ an LTS and
    ${\leadsto} \subseteq \Tr\times\Tr\times\Tr$ the \emph{successor relation}\index{successor}
    such that if $(t,u,v)\in{\leadsto}$ (also denoted by $t\leadsto_{u}v$) then $\source(t)=\source(u)$ and $\source(v)=\target(u)$.
\end{definitionc}

\begin{example}\label{ex:fillipo2}
Remember that the `classical' LTSs of \ex{filippo} are identical.
Let $t_1$ and $t_2$ be the two transitions 
corresponding to \lstinline{y:=y+1} in the first and second state, respectively, and let $u$ be the transition 
for assignment \lstinline{x:=1}. The assignments of \lstinline{x} and \lstinline{y} in the right-hand program are independent, hence
$t_1\leadsto_{u} t_2$ and $u\leadsto_{t_1} u$. For the other program, the situation is different: as the instructions correspond to a single component (program), all transitions affect each other, i.e. ${\leadsto} = \emptyset$.
\end{example}

\begin{definitionc}{Ep-bisimilarity \cite{GHW21ea}}\label{df:ep-bisimilarity}\upshape
Let $(S,\Tr,\source,\target,\ell,\leadsto)$ be an LTSS\@. An \emph{enabling preserving bisimulation (ep-bisimulation)}\index{ep-bisimulation} is a relation
    $\R \subseteq S\times S\times\Pow(\Tr\times\Tr)$ satisfying
    \begin{enumerate}
      \item if $(p,q,R)\in\R$ then $R \subseteq \en(p)\times\en(q)$ such that\label{epb1}
              \begin{enumerate}[label=\alph*{\;\!.},ref=\theenumi.\alph*]
                \item $\forall t\in\en(p).~ \exists\, u\in\en(q).~ t \mathrel{R} u$,\label{epb1a}
                \item $\forall u\in\en(q).~ \exists\, t\in\en(p).~ t \mathrel{R} u$,\label{epb1b} and
                \item if $t \mathrel{R} u$ then $\ell(t)=\ell(u)$;\label{epb1c} and
              \end{enumerate}
      \item if $(p,q,R)\in\R$ and $v \mathrel{R} w$, then $(\target(v),\target(w),R')\in\R$ for some $R'$ such that\label{epb2}
              \begin{enumerate}[label=\alph*\;\!.,ref=\theenumi.\alph*]
                \item if $t \mathrel{R} u$ and $t\leadsto_v t'$
                        then $\exists\, u'.~ u\leadsto_w u' \land t' \mathrel{R'} u'$,\label{epb2a} and
                \item if $t \mathrel{R} u$ and $u\leadsto_w u'$
                        then $\exists\, t'.~ t\leadsto_v t'\land t' \mathrel{R'} u'$.\label{epb2b}
              \end{enumerate}
    \end{enumerate}\pagebreak[3]
    Two states $p$ and $q$ in an LTSS are \emph{enabling preserving bisimilar (ep-bisimilar)}\index{ep-bisimilarity},
    denoted as $p \bisep q$, if there is an enabling preserving bisimulation $\R$ such that $(p,q,R)\mathbin\in\R$ for some $R$.
\end{definitionc}

\noindent
Without Items \ref{epb2a} and \ref{epb2b}, the above is nothing
else than a reformulation of the classical definition of strong bisimilarity. 
An ep-bisimulation additionally maintains for each pair of related states $p$ and $q$ a relation $R$ between the
transitions enabled in $p$ and $q$. 
Items \ref{epb2a} and \ref{epb2b} strengthen the condition on related target states
by requiring that the successors of related transitions are again related relative to these target states.
It is this requirement which distinguishes the transition systems for \ex{filippo}.~\cite{GHW21ea}

\begin{lemma}[Proposition~10 of \cite{GHW21ea}]
$\bisep$ is an equivalence relation.
\end{lemma}

\section{An Introductory Example: CCS with Successors\label{sec:CCS}}
Before starting to introduce the concepts formally, we want to present some motivation in the
form of the well-known Calculus of Communicating Systems (CCS)~\cite{Mi90ccs}.
In this paper we use a proper recursion construct instead of agent identifiers with defining equations.
As in \cite{BW90}, we write $\rec{X|S}$ for the $X$-component of a solution of the set of
  recursive equations $S$.

CCS is parametrised with set ${\Ch}$ of {\em handshake communication names}. 
$\bar{\Ch} \coloneqq \{\bar{c} \mid c\in\Ch\}$ is the set of \emph{handshake communication co-names}.
$\Actc \coloneqq \Ch \djcup \bar{\Ch} \djcup \{\tau\}$
is the set of {\em actions}, where $\tau$ is a special \emph{internal action}.
Complementation extends to $\Ch \djcup \bar{\Ch}$ by $\bar{\bar{c}} \coloneqq c$.

Below, $c$ ranges over $\Ch \djcup \bar{\Ch}$ and
  $\alpha$, $\ell$, $\eta$ over $\Actc$.
  A \emph{relabelling} is a function $f:\Ch\to\Ch$;
  it extends to $\Actc$ by $f(\bar{c})=\overline{f(c)}$, $f(\tau)\coloneqq\tau$.

\begin{table}[tb]
  \vspace{-1.5ex}
  \caption{Structural operational semantics of CCS\label{tab:CCS transition rules}}
  \vspace{1ex}
  \centering
  \framebox{$\begin{array}{ccc}
    \displaystyle\frac{}{\alpha.x \goto{\alpha} x}                                            \mylabelA{actAlpha}{\actsyn{\alpha}} &
    \displaystyle\frac{x \goto{\alpha} x'}{x+y \goto{\alpha} x'}                              \mylabelA{plusL}{+^{}_{\!\!\Left}} &
    \displaystyle\frac{y \goto{\alpha} y'}{x+y \goto{\alpha} y'}                              \mylabelA{plusR}{+^{}_{\!\!\Right}} \\[4ex]
    \displaystyle\frac{x \goto{\eta} x'}{x|y \goto{\eta} x'|y}                                \mylabelA{parL}{|^{}_\Left} &
    \displaystyle\frac{x \goto{c} x',~ y \goto{\bar{c}} y'}{x|y \goto{\tau} x'|y'}            \mylabelA{parH}{|^{}_\mathrm{\scriptscriptstyle C}} &
    \displaystyle\frac{y \goto{\eta} y'}{x|y \goto{\eta} x|y'}                                \mylabelA{parR}{|^{}_\Right} \\[4ex]
    \displaystyle\frac{x \goto{\ell} x'~~\color{Green}(\ell\notin L \djcup \overline{L})}
      {x\backslash L \goto{\ell} x'\backslash L}                                              \mylabelA{restr}{\backslash L} &
    \displaystyle\frac{x \goto{\ell} x'}{x[f] \goto{f(\ell)} x'[f]}                           \mylabelA{relab}{[f]} &
    \displaystyle\frac{\rec{S_X|S} \goto{\alpha} y}{\rec{X|S} \goto{\alpha} y}                \mylabelA{recAct}{rec_{\Act}} \\[3ex]
  \end{array}$}
\end{table}

The process signature $\Sigma$ of CCS features binary infix-written operators $+$ and $|$, denoting \emph{choice} and \emph{parallel composition}, a constant ${\bf 0}$ denoting \emph{inaction}, 
  a unary \emph{action prefixing} operator $\alpha.\_\!\_$ for each action $\alpha\in\Actc$,
  a unary \emph{restriction} operator $\_\!\_{\setminus}L$ for each set $L\subseteq\Ch$, and
  a unary \emph{relabelling} operator $\_\!\_[f]$ for each relabelling
  $f:\Ch\to\Ch$. 

The semantics of CCS is given by the set $\RR$ of \emph{transition rules}, shown in \tab{CCS transition rules}.
Here $\overline{L} \coloneqq \{\bar{c} \mid c\in L\}$.
Each rule has a unique name, displayed in {\color{blue} blue}.\footnote{Our colourings are for readability only.}
The rules are displayed as templates, following
  the standard convention of labelling transitions with \emph{label variables} $c$, $\alpha$, $\ell$, etc.\@
  and may be accompanied by side conditions in {\color{Green} green},
  so that each of those templates corresponds to a set of (concrete) transition rules
  where label variables are ``instantiated'' to labels in certain ranges and all side conditions are met.
The rule names are also schematic and may contain variables.
For example, all instances of the transition rule template $\myrefA{plusL}$ are named $\myrefA{plusL}$, whereas there is one rule name $\myrefA{actAlpha}$ for each action $\alpha\in\Actc$.

The transition system specification $\itt$ is in De Simone format~\cite{dS85}, a special rule format that guarantees properties of the process algebra (for free), such as strong bisimulation being a congruence for all operators.
Following \cite{GHW21ea}, we leave out the infinite sum $\sum_{i \in I} x_i$ of CCS~\cite{Mi90ccs}, as it is strictly speaking not in De Simone format.

In this paper, we will extend the De Simone format to also guarantee properties for ep-bisimulation.
As seen, ep-bisimulation requires that the structural operational semantics is equipped with a successor relation $\leadsto$.
The meaning of $\chi \leadsto_{\zeta} \chi'$ is that transition $\chi$ is unaffected by $\zeta$ -- denoted $\chi \aconc \zeta$ --
and that when doing $\zeta$ instead of $\chi$, afterwards a variant $\chi'$ of $\chi$ is still enabled.
\tab{CCS successor rules} shows the \emph{successor rules} for CCS, which allow the relation $\leadsto$ to be derived inductively.
It uses the following syntax for transitions $\chi$, which will be formally introduced in \Sec{transition expressions}.
The expression $t {\myrefA{plusL}} Q$ refers to the transition that is derived by rule $\myrefA{plusL}$ of \tab{CCS transition rules},
with $t$ referring to the transition used in the unique premise of this rule, and $Q$ referring to
the process in the inactive argument of the $+$-operator. The syntax for the other transitions is
analogous. A small deviation of this scheme occurs for recursion: $\RecAct(X,S,t)$ refers to the
transition derived by rule $\myrefA{recAct}$ out of the premise $t$, when deriving a transition of a
recursive call $\rec{X|S}$. 

In \tab{CCS successor rules} each rule is named, in {\color{orange} orange},
  after the number of the clause of Definition 20 in \cite{GHW21ea}, were it was introduced.

{
\newcommand{\RS}{S}   
\renewcommand{\tx}{t} 
\newcommand{\txp}{t'} 
\newcommand{\uy}{u}   
\newcommand{\uyp}{u'} 
\renewcommand{\vx}{v} 
\newcommand{\wy}{w}   
\newcommand{\set}[1]{{\scriptstyle\{#1\}}}
\newcommand{\name}[1]{{\color{orange}#1}}%
The primary source of concurrency between transition $\chi$ and $\zeta$ is when they stem from
opposite sides of a parallel composition. This is expressed by Rules~\name{7a} and~\name{7b}.
We require all obtained successor statements
$\chi \leadsto_{\zeta} \chi'$ to satisfy the conditions of \df{LTSS} -- this yields
$Q'=\target(w)$ and $P'=\target(v)$; in \cite{GHW21ea} $Q'$ and $P'$ were written this way.
  
In all other cases, successors of $\chi$ are inherited from
successors of their building blocks.

When $\zeta$ stems from the left side of a $+$ via rule $\myrefA{plusL}$ of \tab{CCS transition rules},
then any transition $\chi$ stemming from the right is discarded by $\zeta$, so $\chi \naconc \zeta$.
Thus, if $\chi \aconc \zeta$ then these transitions have the form $\chi=\tx \myrefA{plusL} Q$ and
$\zeta=\vx \myrefA{plusL} Q$, and we must have $\tx \aconc \vx$. So $\tx \leadsto_v \txp$ for some transition $\txp$.
As the execution of $\zeta$ discards the summand $Q$, we also obtain $\chi \leadsto_{\zeta} \txp$.
This motivates Rule~\name{3a}. Rule~\name{4a} follows by symmetry.

In a similar way, Rule~\name{8a} covers the case that $\chi$ and $\zeta$ both stem from the left
component of a parallel composition.
It can also happen that $\chi$ stems form the left component, whereas $\zeta$ is a synchronisation,
involving both components. Thus $\chi=\tx\myrefA{parL} Q$ and $\zeta=\vx \myrefA{parH} \wy$. For $\chi\aconc\zeta$ to hold, it must be
that $\tx\aconc \vx$, whereas the $\wy$-part of $\zeta$ cannot interfere with $t$. This yields the Rule~\name{8b}.
Rule~\name{8c} is explained in a similar vain from the possibility that $\zeta$
stems from the left while $\chi$ is a synchronisation of both components.
Rule~\name{9} follows by symmetry.
In case both $\chi$ and $\zeta$ are synchronisations involving both components, i.e., $\chi=\tx \myrefA{parH} \uy$ and
$\zeta=\vx \myrefA{parH} \wy$, it must be that $\tx \aconc \vx$ and $\uy \aconc \wy$. Now the resulting variant $\chi'$ of
$\chi$ after $\zeta$ is simply $\txp|\uyp$, where $\tx \leadsto_{\vx} \txp$ and $\uy \leadsto_{\wy} \uyp$.
This underpins Rule~\name{10}.

If the common source $\sR$ of $\chi$ and $\zeta$ has the form $P\relab$, $\chi$ and $\zeta$ must have
the form $\tx\myrefA{relab}$ and $\vx\myrefA{relab}$.
Whether $\tx$ and $\vx$ are concurrent is not influenced by the renaming. So $\tx\aconc \vx$.
The variant of $\tx$ that remains after doing $\vx$ is also not affected by the renaming,
so if $\tx \leadsto_{\vx} \txp$ then $\chi \leadsto_{\zeta} \txp\myrefA{relab}$. The case that $\sR = P{\setminus}L$ is equally trivial. This yields Rules~\name{11a} and~\name{11b}.

In case $\sR=\rec{X|\RS}\!$, $\chi$ must have the form $\RecAct(X,S,\tx)$,
and $\zeta$ has the form $\RecAct(X,S,\vx)$,
where $\tx$ and $\vx$ are enabled in $\rec{\RS_X|\RS}$. Now $\chi\mathbin{\aconc} \zeta$ only if $\tx \mathbin{\aconc} \vx$,
so $\tx \mathbin{\leadsto_{\vx}} \txp$ for some transition~$\txp$.
The recursive call disappears upon executing
$\zeta$, and we obtain $\chi \leadsto_{\zeta} \txp$. This yields Rule~\name{11c}.%
}

\begin{table}[t]
  \vspace{-1.5ex}
  \caption{Successor rules for CCS\label{tab:CCS successor rules}}
  \vspace{1ex}
  \centering
  \framebox{$\begin{array}{c@{\ \ \ \ }c@{\ \ \ \ }c}

    \multicolumn{3}{c}{
    \displaystyle\frac{t \leadsto_v t'}{
    t {\myrefA{plusL}} Q \leadsto_{v {\myrefA{plusL}} Q} t'
    }\mylabelS{3a} \qquad\qquad
    \displaystyle\frac{u \leadsto_w u'}{
        P {\myrefA{plusR}} u \leadsto_{P {\myrefA{plusR}} w} u'
    }\mylabelS{4a} }
 \\[4ex]

    \multicolumn{3}{c}{
      \displaystyle\frac{}{t {\myrefA{parL}} Q \leadsto_{P {\myrefA{parR}} w} t {\myrefA{parL}} Q'}\mylabelS{7a} \qquad
      \displaystyle\frac{t \leadsto_v t' \quad u \leadsto_w u'}{t {\myrefA{parH}} u \leadsto_{v {\myrefA{parH}} w} t' {\myrefA{parH}} u'}\mylabelS{10} \qquad
      \displaystyle\frac{}{P {\myrefA{parR}} u \leadsto_{v {\myrefA{parL}} Q} P' {\myrefA{parR}} u}\mylabelS{7b}} \\[4ex]

    \multicolumn{3}{c}{
      \displaystyle\frac{t \leadsto_v t'}{t {\myrefA{parL}} Q \leadsto_{v {\myrefA{parL}} Q} t' {\myrefA{parL}} Q}\mylabelS{8a} \qquad
      \displaystyle\frac{t \leadsto_v t'}{t {\myrefA{parL}} Q \leadsto_{v {\myrefA{parH}} w} t' {\myrefA{parL}} Q'}\mylabelS{8b} \qquad
      \displaystyle\frac{t \leadsto_v t'}{t {\myrefA{parH}} u \leadsto_{v {\myrefA{parL}} Q} t' {\myrefA{parH}} u}\mylabelS{8c}} \\[4ex]

    \multicolumn{3}{c}{
      \displaystyle\frac{u \leadsto_w u'}{P {\myrefA{parR}} u \leadsto_{P {\myrefA{parR}} w} P {\myrefA{parR}} u'}\mylabelS{9a} \qquad
      \displaystyle\frac{u \leadsto_w u'}{P {\myrefA{parR}} u \leadsto_{v {\myrefA{parH}} w} P' {\myrefA{parR}} u'}\mylabelS{9b} \qquad
      \displaystyle\frac{u \leadsto_w u'}{t {\myrefA{parH}} u \leadsto_{P {\myrefA{parR}} w} t {\myrefA{parH}} u'}\mylabelS{9c} } \\[4ex]
    
    \multicolumn{3}{c}{
      \displaystyle\frac{t \leadsto_v t'}{t{\myrefA{restr}} \leadsto_{v{\myrefA{restr}}} t'{\myrefA{restr}}}\mylabelS{11a} \qquad
      \displaystyle\frac{t \leadsto_v t'}{t{\myrefA{relab}} \leadsto_{v{\myrefA{relab}}} t'{\myrefA{relab}}}\mylabelS{11b} \qquad
      \displaystyle\frac{t \leadsto_v t'}{
          \RecAct(X,S,t) \leadsto_{\RecAct(X,S,v)} t'
        }\mylabelS{11c}}
  \end{array}$}
  \vspace{-1ex}
\end{table}

\begin{example}
The programs from~\ex{filippo} could be represented in CCS as
$P \mathbin{:=} \rec{X|S}$ where\linebreak $S={\scriptstyle\left\{\begin{array}{l} X = a.X  + b.Y \\ Y = a.Y\end{array}\!\!\right\}}$ and 
$Q := \rec{Z|\{Z = a.Z\}}| b.{\mathbf{0}}$. Here $a,b\in \Actc$ are the atomic actions incrementing
$y$ and $x$.
The relation matching $P$ with $Q$ and $\rec{Y,S}$ with $\rec{Z|\{Z = a.Z\}}| {\mathbf{0}}$
is a strong bisimulation. Yet, $P$ and $Q$ are not ep-bisimilar,
as the rules of \tab{CCS successor rules} derive $u \leadsto_{t_1} u$ (cf.~\ex{fillipo2})
where \plat{$u = \rec{Z|\{Z = a.Z\}} \myrefA{parR} {\color{blue}\actsyn{b}}\mathbf{0}$} and
\plat{$t_1 = \recAct(Z,\{Z{=}a.Z\},{\color{blue}\actsyn{a}}Q)\myrefA{parL}b.\mathbf{0}$}.
This cannot be matched by $P$, thus violating condition 2.b.\ of \df{ep-bisimilarity}.
\end{example}

\noindent
In this paper we will introduce a new De Simone format for transition systems with successors (TSSS).
We will show that $\bisep$ is a congruence for all operators (as
well as a lean congruence for recursion) in any language that fits this format.
Since the rules of \tab{CCS successor rules} fit this new De Simone format,
it follows that $\bisep$ is a congruence for the operators of CCS.

Informally, the conclusion of a successor rule in this extension of the De Simone format must have
the form $\zeta \leadsto_\xi \zeta'$ where $\zeta$, $\xi$ and $\zeta'$ are \emph{open transitions},
denoted by \emph{transition expressions} with variables, formally introduced in \Sec{transition expressions}.
Both $\zeta$ and $\xi$ must have a leading operator $\snr$ and $\sns$ of the same type, and the same number of arguments.
These leading operators must be rule names of the same type. Their arguments are either process variables $P,Q,...$
or transition variables $t,u,...$, as determined by the trigger sets $I_{\snr}$ and $I_\sns$ of $\snr$ and~$\sns$.
These are the sets of indices listing the arguments for which rules $\snr$ and~$\sns$ have a premise. 
If the $i^{\rm th}$ arguments of $\snr$ and $\sns$ are both process variables, they must be the
same, but for the rest all these variables are different.
For a subset $I$ of $I_\snr \cap I_\sns$, the rule has premises $t_i \leadsto_{u_i} t'_i$ for $i \in I$,
where $t_i$ and $u_i$ are the $i^{\rm th}$ arguments of $\snr$ and $\sns$,
and $t'_i$ is a fresh variable. Finally, the right-hand side of the conclusion may be an arbitrary
univariate transition expression, containing no other variables than:
\begin{itemize}
\item the $t'_i$ for $i\in I$,
\item a $t_i$ occurring in $\zeta$, with $i\notin I_\sns$,
\item a fresh process variable $P'_i$ that must match the target of the transition $u_i$ for $i\in I_\sns{\setminus}I$,
\item \emph{or} a fresh transition variable whose source matches the target of $u_i$ for $i\in I_\sns{\setminus}I$, and
\item any $P$ occurring in both $\zeta$ and $\xi$, \emph{or} any fresh transition variable whose source must be $P$.
\end{itemize}
The rules of \tab{CCS successor rules} only feature the first three possibilities; the others occur
in the successor relation of {\ABCdE} -- see \Sec{abcde}.

\section{Structural Operational Semantics}\label{sec:structural operational semantics}
Both the De Simone format and our forthcoming extension are based on the syntactic form of
  the operational rules. 
In this section,  we recapitulate foundational definitions needed later on.
Let $\Var$ be an infinite set of \emph{process variables}, ranged over by
$X,Y,x,y,x_i$, etc.

\newcommand{\E}{p}

\begin{definitionc}{Process Expressions \cite{vG94a}}\label{df:process}\upshape
  An \emph{operator declaration}\index{operator declaration} is a pair $(\Op,n)$
    of an \emph{operator symbol}\index{operator symbol} $\Op\notin\Var$ and an \emph{arity}\index{arity} $n\in\IN$.
  An operator declaration $(c,0)$ is also called a \emph{constant declaration}\index{constant declaration}.
  A \emph{process signature}\index{process signature} is a set of operator declarations.
  The set $\IP^r(\Sigma)$ of \emph{process expressions}\index{process expression} over a process signature $\Sigma$ is defined inductively by:
  \begin{itemize}
    \item $\Var\subseteq\IP^r(\Sigma)$,
    \item if $(\Op,n)\in\Sigma$ and $p_1,\dots,p_n \in \IP^r(\Sigma)$ then $\Op(p_1,\dots,p_n)\in\IP^r(\Sigma)$, and
    \item if $V_S\subseteq\Var$, $S:V_S\to\IP^r(\Sigma)$ and $X\in V_S$, then $\rec{X|S}\in\IP^r(\Sigma)$.
  \end{itemize}
  A process expression $c()$ is abbreviated as $c$ and is also called a \emph{constant}\index{constant}.
An expression $\rec{X|S}$ as appears in the last clause is called a \emph{recursive call}, and
 the function $S$ therein is called a \emph{recursive specification}\index{recursive specification}.
  It is often displayed as $\{X=S_X \mid X\in V_S\}$.
  Therefore, for a recursive specification $S$, $V_S$ denotes the domain of $S$ and $S_X$ represents $S(X)$ when $X\in V_S$.
  Each expression $S_Y$ for $Y\in V_S$ counts as a subexpression of $\rec{X|S}$.
  An occurrence of a process variable $y$ in an expression $\E$ is \emph{free}\index{free}
    if it does not occur in a subexpression of the form $\rec{X|S}$ with $y\in V_S$.
  For an expression $\E$, $\var(\E)$ denotes the set of process variables having at least one free occurrence in $\E$.
  An expression is \emph{closed}\index{closed} if it contains no free occurrences of variables.
  Let $\cP^r(\Sigma)$ be the set of closed process expressions over $\Sigma$.
\end{definitionc}

\begin{definitionc}{Substitution}\label{df:substitution}\upshape
  A \emph{$\Sigma$-substitution}\index{$\Sigma$-substitution} $\sigma$ is a partial function from $\Var$ to $\IP^r(\Sigma)$.
  It is \emph{closed}\index{closed} if it is a total function from $\Var$ to $\cP^r(\Sigma)$.

  If $p\in\IP^r(\Sigma)$ and $\sigma$ a $\Sigma$-substitution, then $p[\sigma]$ denotes the expression obtained from $p$
    by replacing, for $x$ in the domain of $\sigma$, every free occurrence of $x$ in $p$ by $\sigma(x)$,
    while renaming bound process variables if necessary to prevent name-clashes.
  In that case $p[\sigma]$ is called a \emph{substitution instance}\index{substitution instance} of $p$.
  A substitution instance $p[\sigma]$ where $\sigma$ is given by $\sigma(x_i)=q_i$ for $i\in I$ is denoted as $p[q_i/x_i]_{i\in I}$,
    and for $S$ a recursive specification $\rec{p|S}$ abbreviates $p[\rec{Y|S}/Y]_{Y\in V_S}$.

  These notions, including ``free''\index{free} and ``closed''\index{closed}, extend to syntactic objects containing expressions,
    with the understanding that such an object is a substitution instance of another one
    if the same substitution has been applied to each of its constituent expressions.
\end{definitionc}

\noindent
We assume fixed but arbitrary sets $\Lab$ and $\NN$ of \emph{transition labels} and \emph{rule names}.
\begin{definitionc}{Transition System Specification \cite{GrV92}}\label{df:TSS}\upshape
  Let $\Sigma$ be a process signature.
  A \emph{\lit-(transition) literal}\index{\lit-(transition) literal}
    is an expression \plat{$p \goto{a} q$} with $p,q\in\IP^r(\Sigma)$ and $a\mathbin\in\Lab$.
  A \emph{transition rule}\index{transition rule} over \lit is an expression of the form $\frac{H}{\lambda}$ with
    $H$ a finite list of \lit-literals (the \emph{premises}\index{premise} of the transition rule) and
    $\lambda$ a \lit-literal (the \emph{conclusion}\index{conclusion}).
  A \emph{transition system specification (TSS)}\index{TSS} is a tuple \tss 
  with $\RR$ a set of transition rules over \lit, and $\scn:\RR\to\NN$ a (not necessarily injective) \emph{rule-naming function},
  that provides each rule $\scr\in\RR$ with a name $\scn(\scr)$.
\end{definitionc}

\begin{definitionc}{Proof}\label{df:proof}\upshape
  Assume literals, rules, substitution instances and rule-naming. A
    \emph{proof}\index{proof} of a literal $\lambda$ from a set $\RR$ of rules is a well-founded,
    upwardly branching, ordered tree where nodes are labelled by pairs $(\mu,\snr)$ of a
    literal $\mu$ and a rule name $\snr$, such that
    \begin{itemize}
      \item the root is labelled by a pair $(\lambda,\sns)$, and
      \item if $(\mu,\snr)$ is the label of a node and
              $(\mu_1,\snr_1),\dots,(\mu_n,\snr_n)$ is the list of labels of this node's children
              then $\frac{\mu_1,\dots,\mu_n}{\mu}$ is a substitution instance of a rule in $\RR$ with name $\snr$.
    \end{itemize}
\end{definitionc}

\begin{definitionc}{Associated LTS \cite{GH19}}\label{df:associated LTS}\upshape
  The \emph{associated LTS}\index{associated LTS} of a TSS \tss is the LTS $(S,\Tr,\source,\linebreak[1]\target,\ell)$ with
    $S \coloneqq \cP^r(\Sigma)$ and $\Tr$ the collection of proofs $\pi$ of closed \lit-literals \plat{$p \goto{a} q$} from $\RR$,
    where $\source(\pi)=p$, $\ell(\pi)=a$ and $\target(\pi)=q$.
\end{definitionc}

\noindent
Above we deviate from the standard treatment of structural operational semantics \cite{GrV92,vG94a}
on four counts. Here we employ {CCS} to motivate those design decisions.

In \df{associated LTS}, the transitions $\Tr$ are taken to be proofs of closed literals $p\goto{a} q$
rather than such literals themselves. This is because there can be multiple $a$-transitions
from $p$ to $q$ that need to be distinguished when taking the concurrency relation between
transitions into account. For example, if $p:= \rec{X|\{X=a.X+c.X\}}$ and $q := \rec{Y|\{Y=a.Y\}}$
then $p|q$ has three outgoing transitions:

\vspace{-2ex}
\mprset{rightstyle={\scriptsize}} 
{\footnotesize
\begin{mathpar}
 \inferrule*[Right=$\myrefA{parL}$]
 	{\inferrule*[Right=$\myrefA{recAct}$]
		{\inferrule*[Right=$\myrefA{plusL}$]
			{\inferrule*[Right=$\color{blue}{\actsyn a}$]
				{ }
				{a.p \goto{a} p}
			}
			{a.p+c.p \goto{a} p}
		}
	  	{p \goto{a} p}
	}
	{p|q \goto{a} p|q}

 \inferrule*[Right=$\myrefA{parL}$]
 	{\inferrule*[Right=$\myrefA{recAct}$]
		{\inferrule*[Right=$\myrefA{plusR}$]
			{\inferrule*[Right=$\color{blue}{\actsyn c}$]
				{ }
				{c.p \goto{c} p}
			}
			{a.p+c.p \goto{c} p}
		}
	  	{p \goto{c} p}
	}
	{p|q \goto{c} p|q}

 \inferrule*[Right=$\myrefA{parR}$]
 	{\inferrule*[Right=$\myrefA{recAct}$]
		{\inferrule*[Right=$\color{blue}{\actsyn a}$]
			{ }
			{a.q \goto{a} q}
		}
	  	{q \goto{a} q}
	}
	{p|q \goto{a} p|q}
	\vspace{-1.5ex}
\end{mathpar}}
The rightmost transition is concurrent with the middle one, whereas the leftmost one is not.

A similar example can be used to motivate why in \df{proof} the nodes are labelled not only by the
inferred literal, but also by the name of the applied rule.

\vspace{-4mm}
{\footnotesize
\begin{mathpar}
 \inferrule*[Right=$\myrefA{parL}$]
 	{\inferrule*[Right=$\myrefA{recAct}$]
		{\inferrule*[Right=$\myrefA{plusL}$]
			{\inferrule*[Right=$\color{blue}{\actsyn a}$]
				{ }
				{a.p \goto{a} p}
			}
			{a.p+c.p \goto{a} p}
		}
	  	{p \goto{a} p}
	}
	{p|p \goto{a} p|p}

 \inferrule*[Right=$\myrefA{parL}$]
 	{\inferrule*[Right=$\myrefA{recAct}$]
		{\inferrule*[Right=$\myrefA{plusR}$]
			{\inferrule*[Right=$\color{blue}{\actsyn c}$]
				{ }
				{c.p \goto{c} p}
			}
			{a.p+c.p \goto{c} p}
		}
	  	{p \goto{c} p}
	}
	{p|p \goto{c} p|p}

 \inferrule*[Right=$\myrefA{parR}$]
 	{\inferrule*[Right=$\myrefA{recAct}$]
		{\inferrule*[Right=$\myrefA{plusL}$]
			{\inferrule*[Right=$\color{blue}{\actsyn a}$]
				{ }
				{a.p \goto{a} p}
			}
			{a.p+c.p \goto{a} p}
		}
	  	{p \goto{a} p}
	}
	{p|p \goto{a} p|p}
	\vspace{-2ex}
\end{mathpar}}%
The rightmost transition is concurrent with the middle one, but the leftmost one is not.
If we were to erase the rule names, the difference between these two transitions would disappear.

In \df{TSS} we require the premises of rules to be lists rather than sets, and accordingly in
\df{proof} we require proof trees to be ordered. This is to distinguish transitions/proofs
in which a substitution instance of a rule has two identical premises (corresponding to different
arguments of the leading operator) with different proofs.
This phenomenon does not occur in CCS, but we could have illustrated it with CSP~\cite{BHR84} or ABCdE~\cite{GHW21ea}.

{\newcommand{\Nil}{\mathbf{0}}
Finally, suppose that in \df{TSS} we had chosen the rule-naming function $\scn$ to be the identity.
This is equivalent to not having a rule-naming function at all, instead labelling nodes in proofs
with rules rather than names of rules. Then in the transition

\vspace{-3.3ex}
{\footnotesize
\begin{mathpar}
 \inferrule*[Right=$\myrefA{recAct}$]
 	{\inferrule*[Right=$\color{blue}{\actsyn a}$]
		{ }
	  	{a.\Nil \goto{a} \Nil}
	}
	{\rec{X|\{X=a.\Nil} \goto{a} \Nil}
	\vspace{-2.5mm}
\end{mathpar}}%
we should replace the generic name $\myrefA{recAct}$ of a recursion rule with the specific rule employed.
This could be the rule 
{\scriptsize$\inferrule{a.\Nil \goto{a} z}{\rec{X|\{X=a.\Nil\}} \goto{a} z}$}, 
but just as well the rule 
{\scriptsize$\inferrule{y \goto{a} z}{\rec{X|\{X=y\}} \goto{a} z}$},
when employing a substitution that sends $y$ to $a.\Nil$.
To avoid the resulting unnecessary duplication of transitions,
we give both recursion rules the same name.}

\section{De Simone Languages}\label{sec:De Simone languages}

The syntax of a \emph{De Simone language}\index{De Simone language} is specified by a process signature, and its semantics is given as a TSS over that process signature of a particular form \cite{dS85}, nowadays known as the \emph{De Simone format}.
    Here, we extend the De Simone format to support \emph{indicator transitions}, as occur in
    \cite{GH15a,vG19,GHW21ea}. \pagebreak[3] These are transitions \plat{$p \goto{\ell} q$} for which it is
    essential that $p=q$. They are used to convey a property of the state $p$ rather than model an
    action of $p$. To accommodate them we need a variant of the recursion rule whose
    conclusion again is of the form \plat{$r \goto{\ell} r$}.
   This variant will be illustrated in \Sec{abcde}.

As for $\Lab$, we fix a set $\Act\subseteq\Lab$ of \emph{actions}.

\begin{definitionc}{De Simone Format}\label{df:TSS in De Simone format}\upshape
  A TSS \tss is in \emph{De Simone format}\index{De Simone format} if
    for every recursive call $\rec{X|S}$ and every $\alpha\in \Act$ and
      $\ell \in \Lab{\setminus}\Act$, it has transition rules
    \[
      \frac{\rec{S_X|S} \goto{\alpha} y}{\rec{X|S} \goto{\alpha} y}~\myrefA{recAct}
 \qquad\text{and}\qquad
      \frac{\rec{S_X|S} \goto{\ell} y}{\rec{X|S} \goto{\ell} \rec{X|S}}~\mylabelA{recIn}{rec_{\In}}
 \quad\text{for some}\quad y\notin\var(\rec{S_X|S}),
    \]
    and each of its other transition rules (\emph{De Simone rules}) has the form
    \[
      \frac{\{x_i \goto{a_i} y_i \mid i\in I\}}{\Op(x_1,\dots,x_n) \goto{a} q}
    \]
    where $(\Op,n)\in\Sigma$, $I\subseteq\{1,\dots,n\}$, $a,a_i\in\Lab$, $x_i$ (for $1\leq i\leq n$) and $y_i$ (for $i\in I$)
    are pairwise distinct process variables, and $q$ is a univariate process expression
    containing no other free process variables than $x_i$ ($1\leq i\leq n \land i\notin I$) and $y_i$ ($i\in I$),
    having the properties that 
      \begin{itemize}
        \item each subexpression of the form $\rec{X|S}$ is closed, and
        \item if $a\in\Lab{\setminus}\Act$ then $a_i\in\Lab{\setminus}\Act$ ($i\in I$) and $q=\Op(z_1,\dots,z_n)$,
          where \raisebox{0pt}[0pt]{\(z_i:=
            \left\{\begin{array}{@{}cl}y_i & \textrm{~if~} i\in I\\
                                       x_i & \textrm{~otherwise}.
                   \end{array}\right.\)}
      \end{itemize}
  Here \emph{univariate} means that each variable has at most one free occurrence in it.
  The last clause above guarantees that for any indicator transition $t$, one with $\ell(t)\in\Lab{\setminus}\Act$,
  we have $\target(t) = \source(t)$. 
  For a De Simone rule of the above form, $n$ is the \emph{arity}\index{arity}, $(\Op,n)$ is the \emph{type}\index{type}, $a$ is the \emph{label}\index{label},
    $q$ is the \emph{target}\index{target}, $I$ is the \emph{trigger set}\index{trigger set} and
    the tuple $(\ell_i,\dots,\ell_n)$ with $\ell_i=a_i$ if $i\in I$ and $\ell_i=*$ otherwise, is the \emph{trigger}\index{trigger}.
  Transition rules in the first two clauses are called \emph{recursion rules}.

We also require that if $\scn(\scr)\mathbin=\scn(\scr')$ for two different De Simone rules
$\scr,\scr'\mathbin\in \RR$, then $\scr,\scr'$ have the same type, target and trigger set, but different triggers.
The names of the recursion rules are as indicated in {\color{blue}blue} above, and differ from
the names of any De Simone rules.
\end{definitionc}
Many process description languages encountered in the literature, including
CCS~\cite{Mi90ccs} as presented in \Sec{CCS}, SCCS~\cite{M83}, ACP~\cite{BW90} and \textsc{Meije}~\cite{AB84}, are De Simone languages.

\section{Transition System Specifications with Successors}\label{sec:transition expressions}

In \Sec{structural operational semantics}, a \emph{process} is denoted by a closed process
  expression; an open process expression may contain variables, which stand for as-of-yet unspecified
  subprocesses. Here we will do the same for transition expressions with variables. However,
  in this paper a transition is defined as a proof of a literal $p \goto{a} q$ from the operational
  rules of a language. Elsewhere, a transition is often defined as a provable literal $p \goto{a} q$,
  but here we need to distinguish transitions based on these proofs, as this influences whether two
  transitions are concurrent.

  It turns out to be convenient to introduce an \emph{open proof} of a literal as the semantic
  interpretation of an open transition expression. It is simply a proof in which certain subproofs are replaced by
  proof variables.

\begin{definitionc}{Open Proof}\label{df:open proof}\upshape
  Given definitions of literals, rules and substitution instances, and a rule-naming function $\scn$,
    an \emph{open proof}\index{open proof} of a literal $\lambda$ from a set $\RR$ of rules using a set $\VV$ of \emph{(proof) variables} is a well-founded,
    upwardly branching, ordered tree of which the nodes are labelled either by pairs $(\mu,\snr)$ of a
    literal $\mu$ and a rule name $\snr$, or by pairs $(\mu,\px)$ of a literal $\mu$ and a variable $\px\in\VV$ such that
    \begin{itemize}
      \item the root is labelled by a pair $(\lambda,\chi)$,
      \item if $(\mu,\px)$ is the label of a node then this node has no children,
      \item if two nodes are labelled by $(\mu,\px)$ and $(\mu',\px)$ separately then $\mu=\mu'$, and
      \item if $(\mu,\snr)$ is the label of a node and
              $(\mu_1,\chi_1),\dots,(\mu_n,\chi_n)$ is the list of labels of this node's children
              then $\frac{\mu_1,\dots,\mu_n}{\mu}$ is a substitution instance of a rule named $\snr$.\pagebreak[3]
    \end{itemize}
\end{definitionc}
Let $\TVar$ be an infinite set of \emph{transition variables}, disjoint from $\Var$.
We will use $\tx,\ux,\vx,\ty,\tx_i$, etc.\@ to range over $\TVar$.

\begin{definitionc}{Open Transition}\label{df:open transition}\upshape
  Fix a TSS \tss.
  An \emph{open transition}\index{open transition} is an open proof of a \lit-literal from $\RR$ using $\TVar$.
  For an open transition $\mathring{t}$, $\Tvar(\mathring{t})$ denotes the set of transition variables occurring in $\mathring{t}$;
    if its root is labelled by \plat{$(p \goto{a} q,\chi)$} then $\Osrc(\mathring{t})=p$, $\Oell(\mathring{t})=a$ and $\Otar(\mathring{t})=q$.
    The \emph{binding function}\index{binding function} $\beta_{\mathring{t}}$ of $\mathring{t}$
      from $\Tvar(\mathring{t})$ to \lit-literals
     is defined by $\beta_{\mathring{t}}(\tx)=\mu$
      if $\tx\in\Tvar(\mathring{t})$ and $(\mu,\tx)$ is the label of a node in $\mathring{t}$.
  Given an open transition, we refer to the subproofs obtained by deleting the root node as its \emph{direct subtransitions}\index{direct subtransition}.

  All occurrences of transition variables are considered \emph{free}\index{free}.
  Let $\IT^r\tss$ be the set of open transitions in the TSS \tss
    and $\cT^r\tss$ the set of closed open transitions.
  We have $\cT^r\tss=\Tr$.

  Let $\Oen(p)$ denote $\{\mathring{t} \mid \Osrc(\mathring{t})=p\}$.
\end{definitionc}

\begin{definitionc}{Transition Expression}\label{df:transition expression}\upshape
  A \emph{transition declaration} is a tuple $(\snr,n,I)$ of a \emph{transition constructor} $\snr$,
    an arity $n\in\IN$ and a trigger set $I\subseteq\{1,\dots,n\}$.
  A \emph{transition signature} is a set of transition declarations.
  The set $\ITE^r\ite$ of \emph{transition expressions}
    over a process signature $\Sigma_\PP$ and a transition signature $\Sigma_\TT$
    is defined inductively as follows.
    \begin{itemize}
      \item if $\tx\in\TVar$ and $\mu$ is a \lit-literal then $(\tx \dblcolon \mu)\in\ITE^r\ite$,
      \item if $E\in\ITE^r\ite$, $S:\Var\rightharpoonup\IP^r(\Sigma_\PP)$
        and $X\in\textrm{dom}(S)$\\ 
              then $\RecAct(X,S,E),\RecIn(X,S,E)\in\ITE^r\ite$, and
      \item if $(\snr,n,I)\in\Sigma_\TT$, $E_i\in\ITE^r\ite$ for each $i\in I$,
              and $E_i\in\IP^r(\Sigma_\PP)$ for each $i\in\{1,\dots,n\}\setminus I$, then $\snr(E_1,\dots,E_n)\in\ITE^r\ite$.
    \end{itemize}
  Given a TSS $(\Sigma,\RR,\scn)$ in De Simone format,
  each open transition $\mathring{t}\in\IT^r\itt$ is named by a unique transition expression in $\ITE^r\its$; here
    $\Sigma_\TT = \{(\scn(\scr),n,I) \mid \scr\in\RR$ is a De Simone rule,
    $n$ is its arity and $I$ is its trigger set$\}$:
    \begin{itemize}
      \item if the root of $\mathring{t}$ is labelled by $(\mu,\tx)$ where $\tx\in\TVar$
              then $\mathring{t}$ is named $(\tx \dblcolon \mu)$,
      \item if the root of $\mathring{t}$ is labelled by \plat{$(\rec{X|S} \goto{a} q,\snr)$} where $a\in\Act$
                then $\mathring{t}$ is named $\RecAct(X,S,E)$
                where $E$ is the name of the direct subtransition of $\mathring{t}$,
      \item if the root of $\mathring{t}$ is labelled by \plat{$(\rec{X|S} \goto{\ell} \rec{X|S},\snr)$} where $\ell\in\Lab{\setminus}\Act$
                then $\mathring{t}$ is named $\RecIn(X,S,E)$
                where $E$ is the name of the direct subtransition of $\mathring{t}$, and
      \item if the root of $\mathring{t}$ is labelled by \plat{$(\Op(p_1,\dots,p_n) \goto{a}
        q,\snr)$} then $\mathring{t}$ is named $\snr(E_1,\dots,E_n)$
              where, letting $n$ and $I$ be the arity and the trigger set of the rules named $\snr$,
              $E_i$ for each $i\in I$ is the name of the direct subtransitions of $\mathring{t}$ corresponding to the index $i$,
              and $E_i=p_i$ for each $i\in\{1,\dots,n\}\setminus I$.
    \end{itemize}
\end{definitionc}
We now see that the first requirement for the rule-naming function in \df{TSS in De Simone format}
ensures that every open transition is uniquely identified by its name.

\begin{definitionc}{Transition Substitution}\label{df:transition substitution}\upshape
 Let  \tss be a TSS\@.
  A \emph{$\itt$-substitution}\index{$\itt$-substitution} is a partial function
    $\Tsigma:(\Var\rightharpoonup\IP^r(\Sigma))\cup(\TVar\rightharpoonup\IT^r\itt)$.
  It is \emph{closed}\index{closed} if it is a total function
  $\Tsigma:(\Var\to\cP^r(\Sigma))\cup(\TVar\to\cT^r\itt)$.
  A $\itt$-substitution $\Tsigma$ \emph{matches}\index{match} all process expressions.
  It matches an open transition $\mathring{t}$ whose binding function is $\beta_{\mathring{t}}$
    if for all $(\tx,\mu)\mathbin\in\beta_{\mathring{t}}$,
    $\Tsigma(\tx)$ being defined and \plat{$\mu=(p \goto{a} q)$}
    implies $\Oell(\Tsigma(\tx))=a$ and $\Osrc(\Tsigma(\tx)),\Otar(\Tsigma(\tx))$ being the substitution instances of $p,q$ respectively
    by applying $\Tsigma\mathord\upharpoonright \Var$.
  
  If $E\in\IP^r(\Sigma)\cup\IT^r\itt$ and $\Tsigma$ is a $\itt$-substitution matching $E$,
    then $E[\Tsigma]$ denotes the expression obtained from $E$
    by replacing, for $\tx\in\TVar$ in the domain of $\Tsigma$, every subexpression of the form $(\tx \dblcolon \mu)$ in $E$ by $\Tsigma(\tx)$, and
    for $x\in\Var$ in the domain of $\Tsigma$, every free occurrence of $x$ in $E$ by $\Tsigma(x)$,
    while renaming bound process variables if necessary to prevent name-clashes.
  In that case $E[\Tsigma]$ is called a \emph{substitution instance}\index{substitution instance} of $E$.
\end{definitionc}
Note that a substitution instance of an open transition can be a transition expression not representing an open transition.
For example, applying a $\itt$-substitution $\Tsigma$ given by $\Tsigma(\ty) \coloneqq \plat{$(\tx \dblcolon y \goto{\bar{c}} y')$}$
  to the open transition $(\tx \dblcolon x \goto{c} x') \mid (\ty \dblcolon y \goto{c} y')$
  results in $(\tx \dblcolon x \goto{c} x') \mid (\tx \dblcolon y \goto{c} y')$ which is not an open transition
  because the transition variable $\tx$ is used for two different \lit-literals.
This will not happen if $\Tsigma$ is closed.

\begin{observation}\label{obs:transition composition}\upshape
 Given a TSS \tss, if
    $\mathring{t}\in\Oen(p)$ is a open transition and $\Tsigma$ is a closed $\itt$-substitution
    which matches $\mathring{t}$ then $\mathring{t}[\Tsigma]\in\Tr$,
    $\source(\mathring{t}[\Tsigma])=\Osrc(\mathring{t})[\Tsigma]$,
    $\ell(\mathring{t}[\Tsigma])=\Oell(\mathring{t})$
    and\\ $\target(\mathring{t}[\Tsigma])=\Otar(\mathring{t})[\Tsigma]$.
\end{observation}

\begin{definitionc}{Transition System Specification with Successors}\label{df:TSSS}\upshape
  Let $\tss$ be a TSS\@.
  A \emph{$\itt$-(successor) literal}\index{$\itt$-(successor) literal}
    is an expression \plat{$\mathring{t} \leadsto_{\mathring{u}} \mathring{v}$} with $\mathring{t},\mathring{u},\mathring{v}\in\IT^r\itt$,
    $\Osrc(\mathring{t})=\Osrc(\mathring{u})$ and $\Osrc(\mathring{v})=\Otar(\mathring{u})$.
  A \emph{successor rule}\index{successor rule} over $\itt$ is an expression of the form $\frac{H}{\lambda}$ with
    $H$ a finite list of $\itt$-literals (the \emph{premises}\index{premise} of the successor rule) and
    $\lambda$ a $\itt$-literal (the \emph{conclusion}\index{conclusion}).
  A \emph{transition system specification with successors (TSSS)} is a tuple $\tsss$
    with \tss a TSS and $\UU$ a set of successor rules over $\itt$.
\end{definitionc}

\begin{definitionc}{Associated LTSS}\label{df:associated LTSS}\upshape
For a TSSS \tsss, the \emph{associated LTSS}\index{associated LTSS} 
  is the LTSS $(S,\Tr,\linebreak[1]\source,\target,\ell,\leadsto)$ with
    $S \coloneqq \cP^r(\Sigma)$, $\Tr$ the collection of proofs $\pi$ of closed \lit-literals \plat{$p \goto{a} q$} from $\RR$,
    where $\source(\pi)=p$, $\ell(\pi)=a$ and $\target(\pi)=q$,
    and\\\centerline{${\leadsto} \coloneqq \{(t,u,v) \mid \text{a proof of closed $\itt$-literal } t\leadsto_uv \text{ from } \UU \text{ exists}\}$.}
\end{definitionc}

\section{De Simone Languages with Successors}\label{sec:De Simone languages with successors}
We have enriched standard definitions such as transitions systems and specifications with successors.
This allows up to add successors to the De Simone format to define a new congruence format.

\begin{definitionc}{De Simone Format}\label{df:TSSS in De Simone format}\upshape
  A TSSS $\tsss$ is in \emph{De Simone format}\index{De Simone format} if
    \tss is in De Simone format,
    for every recursive call $\rec{X|S}$ and $\mathit{xa},\mathit{ya},\mathit{za}\in\Lab$ it has a successor rule
    \[
      \frac{(\tx \dblcolon S_X \goto{\mathit{xa}} x') \leadsto_{(\ty \dblcolon S_X \goto{\mathit{ya}} y')} (\tz \dblcolon y' \goto{\mathit{za}} z')}{\mathit{rec}_\chi(X,S,\tx \dblcolon S_X \goto{\mathit{xa}} x') \leadsto_{\RecAct(X,S,\ty \dblcolon S_X \goto{\mathit{ya}} y')}\mathring{u}}
    \]
    where 
    $\mathring{u}= (\tz \dblcolon y' \goto{\mathit{za}} z')$ if $\mathit{ya}\in\Act$
      and $\mathring{u}=\mathit{rec}_\chi(X,S,\tx \dblcolon S_X\goto{\mathit{xa}} x')$ otherwise, 
    $\mathit{rec}_\chi=\RecAct$ if $\mathit{xa}\in\Act$ and $\mathit{rec}_\chi=\RecIn$ otherwise, $x',y',z'$ are pairwise distinct process variables not occurring in $\rec{X|S}$, and $\tx,\ty,\tz$ are pairwise distinct transition variables.
    Moreover, each of its other successor rules has the form
    \[
      \frac{\{(\tx_i \dblcolon x_i \goto{\mathit{xa}_i} x'_i) \leadsto_{(\ty_i \dblcolon x_i \goto{\mathit{ya}_i} y'_i)} (\tz_i \dblcolon y'_i \goto{\mathit{za}_i} z'_i) \mid i\in I\}}
           {\snr(\mathit{xe}_1,\dots,\mathit{xe}_n) \leadsto_{\sns(\mathit{ye}_1,\dots,\mathit{ye}_n)} \mathring{v}}
    \]
such that
    \begin{itemize}
      \item $I\subseteq\{1,\dots,n\}$,
      \item $x_i,x'_i,y'_i,z'_i$ for all relevant $i$ are pairwise distinct process variables,
      \item $\tx_i,\ty_i,\tz_i$ for all relevant $i$ are pairwise distinct transition variables,
      \item if $i\in I$ then $\mathit{xe}_i=(\tx_i \dblcolon x_i \goto{\mathit{xa}_i} x'_i)$ and $\mathit{ye}_i=(\ty_i \dblcolon x_i \goto{\mathit{ya}_i} y'_i)$,
      \item if $i\notin I$ then $\mathit{xe}_i$ is either $x_i$ or $(\tx_i \dblcolon x_i \goto{\mathit{xa}_i} x'_i)$, and
                                $\mathit{ye}_i$ is either $x_i$ or $(\ty_i \dblcolon x_i \goto{\mathit{ya}_i} y'_i)$,
      \item $\snr$ and $\sns$ are $n$-ary transition constructors such that the open
        transitions
$\snr(\mathit{xe}_1,\dots,\mathit{xe}_n)$,\\ $\sns(\mathit{ye}_1,\dots,\mathit{ye}_n)$ and
        $\mathring{v}$ satisfy\\[1mm]
              \mbox{}\hfill
              $\Osrc(\snr(\mathit{xe}_1,\dots,\mathit{xe}_n))=\Osrc(\sns(\mathit{ye}_1,\dots,\mathit{ye}_n))$\hfill\mbox{}\\[1mm] and
              $\Osrc(\mathring{v})=\Otar(\sns(\mathit{ye}_1,\dots,\mathit{ye}_n))$,
      \item $\mathring{v}$ is univariate and contains no other variable
          expressions than 
              \begin{itemize}
                \item $x_i$ or $(\tz_i \dblcolon x_i \goto{\mathit{za}_i} z'_i)$ ($1\leq i\leq n \land \mathit{xe}_i=\mathit{ye}_i=x_i$),
                \item $(\tx_i \dblcolon x_i \goto{\mathit{xa}_i} x'_i)$ ($1\leq i\leq n \land \mathit{xe}_i\neq x_i \land \mathit{ye}_i=x_i$),
                \item $y'_i$ or $(\tz_i \dblcolon y'_i \goto{\mathit{za}_i} z'_i)$ ($1\leq i\leq n \land i\notin I \land \mathit{ye}_i\neq x_i$),
                \item $(\tz_i \dblcolon y'_i \goto{\mathit{za}_i} z'_i)$ ($i \in I$), and
              \end{itemize}
      \item if $\Oell(\sns(\mathit{ye}_1,\dots,\mathit{ye}_n))\in\Lab{\setminus}\Act$ then 
              for $i\in I$, $\mathit{ya}_i\in\Lab{\setminus}\Act$;
              for $i\notin I$, either $\mathit{xe}_i=x_i$ or $\mathit{ye}_i=x_i$;
              and $\mathring{v}=\snr(\mathit{ze}_1,\dots,\mathit{ze}_n)$, where
                \[\raisebox{0pt}[0pt]{\(\mathit{ze}_i:=
                  \left\{\begin{array}{@{}cl} (\tz_i \dblcolon y'_i \goto{\mathit{za}_i} z'_i) & \textrm{~if~} i\in I\\
                                              \mathit{xe}_i & \textrm{~if~} i\notin I \textrm{~and~} \mathit{ye}_i=x_i\\
                                              y'_i & \textrm{~otherwise}.
                \end{array}\right.\)}\]
    \end{itemize}
\end{definitionc}
\begin{figure}[b!]
  \input{inclusions}
  \centerline{\box\graph}
  \caption{Inclusion between index sets $I,I_{\snr},I_{\sns},I_{\textsc t},I_G \subseteq \{1,..,n\}$.
    One has $(I_{\snr}\cap I_G){\setminus}I_{\sns} \subseteq I_{\textsc t}$.}

  \parindent 29pt
    The annotations ${\color{orange}n}_i$ show the location of index $i$ (suppressed for unary
    operators) of rule ${\color{orange}n}$.
  \label{fig:inclusions}
\end{figure}
The last clause above is simply to ensure that if $t \leadsto_{u} v$ for an indicator transition $u$, that is,
  with $\ell(u)\notin \Act$, then $v=t$.

  The other conditions of \df{TSSS in De Simone format} are illustrated by the Venn diagram of \fig{inclusions}.
  The outer circle depicts the indices $1,\dots,n$ numbering the arguments of the operator $\Op$ that is the common
  type of the De Simone rules named $\snr$ and $\sns$; $I_\snr$ and $I_\sns$ are the trigger sets of $\snr$ and $\sns$, respectively.
  In line with \df{transition expression}, $\mathit{xe} = x_i$ for $i \in I_\snr$, and
  $\mathit{xe} = (\tx_i \dblcolon x_i \goto{\mathit{xa}_i} x'_i)$ for $i \notin I_\snr$.
  Likewise, $\mathit{ye} = x_i$ for $i \in I_\sns$, and
  \plat{$\mathit{ye} = (\ty_i \dblcolon x_i \goto{\mathit{xa}_i} y'_i)$} for $i \notin I_\sns$.
  So the premises of any rule named $\sns$ are \plat{$\{x_i \goto{\mathit{xa}_i} y'_i \mid i\in I_\sns\}$}.
  By \df{TSS in De Simone format} the target of such a rule is a univariate process expression $q$ with no other variables
  than $z_1,\dots,z_n$, where $z_i:=x_i$ for $i \in I_\sns$ and $z_i := y'_i$ for $i \notin I_\sns$.
  Since $\Osrc(\mathring{v}) = q$, the transition expression $\mathring{v}$ must be univariate, and have no
  variables other than $\mathit{ze}_i$ for $i=1,\dots,n$, where $\mathit{ze}_i$ is either the process variable $z_i$ or a
  transition variable expression $(\tz_i \dblcolon z_i \goto{\mathit{xa}_i} z'_i)$.

  $I$ is the set of indices $i$ for which the above successor rule has a premise. Since this premise involves the transition
  variables $\tx_i$ and $\ty_i$, necessarily $I \subseteq I_\snr \cap I_\sns$. Let $I_G$ be the set of indices for which
  $\mathit{ze}_i$ occurs in $\mathring{v}$, and $I_{\textsc t} \subseteq I_G$ be the subset where $\mathit{ze}_i$ is a
  transition variable. The conditions on $\mathring{v}$ in \df{TSSS in De Simone format} say that
  $I\cap I_G \subseteq I_{\textsc t}$ and $(I_{\snr}\cap I_G){\setminus}I_{\sns} \subseteq I_{\textsc t}$.
  For $i \in I\cap I_G$, the transition variable $\tz_i$ is inherited from the premises of the rule,
  and for $i \in (I_{\snr}\cap I_G){\setminus}I_{\sns}$ the transition variable $\tz_i$ is inherited from its source.

  In order to show that most classes of indices allowed by our format are indeed populated, we indicated the positions
  of the indices of the rules of CCS and (the forthcoming) \ABCdE from Tables~\ref{tab:CCS successor rules} and
  \ref{tab:ABCdE successor rules}.

Any De Simone language, including  CCS, SCCS, ACP and \textsc{Meije}, can trivially be extended to a
language with successors, e.g.\ by setting $\UU=\emptyset$. This would formalise the assumption that the parallel composition operator of these languages is
governed by a \emph{scheduler}, scheduling actions from different components in a nondeterministic way.
The choice of $\UU$ from \tab{CCS successor rules} instead formalises the assumption that parallel
components act independently, up to synchronisations between them.

We now present the main theorem of this paper, namely that ep-bisimulation is a lean congruence for
all languages that can be presented in De Simone format with successors.
A lean congruence preserves equivalence when replacing closed subexpressions of a process by
equivalent alternatives. Being a lean congruence implies being a congruence for all operators of the
language, but also covers the recursion construct.

\begin{theoremc}{Lean Congruence}\label{thm:lean congruence}\upshape
Ep-bisimulation is a lean congruence for all De Simone languages with successors. Formally, fix a TSSS $\tsss$ in De Simone format.
  If $p\in\IP^r(\Sigma)$ and $\rho,\nu$ are two closed $\Sigma$-substitutions with $\forall x\in\Var.$ $\rho(x) \bisep \nu(x)$
    then $p[\rho] \bisep p[\nu]$.
\end{theoremc}
The proof can be found in \app{a congruence result}.

In contrast to a lean congruence, a full congruence would also allow replacement within a recursive specification of subexpressions that may contain recursion variables bound outside of these subexpressions.
As our proof is already sophisticated, we consider the proof of full congruence to be beyond the scope of the paper.
In fact we are only aware of two papers that provide a proof of full congruence via a rule format~\cite{Re00,vG17b}.

We carefully designed our De Simone format with successors and can state the following conjecture.
\begin{conjecture}
Ep-bisimulation is a full congruence for all De Simone languages with successors.
\end{conjecture}

\section{A Larger Case Study: The Process Algebra \ABCdE\label{sec:abcde}}
The \emph{Algebra of Broadcast Communication with discards and Emissions} (\ABCdE) stems from \cite{GHW21ea}.
It combines CCS~\cite{Mi90ccs}, its extension with broadcast communication~\cite{Prasad91,GH15a,vG19},
  and its extension with signals~\cite{Be88b,CDV09,EPTCS255.2,vG19}.
Here, we extend CCS as presented in \Sec{CCS}.

\ABCdE is parametrised with sets
  $\Ch$ of \emph{handshake communication names} as used in CCS,
  $\B$ of \emph{broadcast communication names} and
  $\Sig$ of \emph{signals}.
  $\bar{\Sig} \coloneqq \{\bar{s} \mid s\in\Sig\}$ is the set of signal emissions.
The collections $\B!$, $\B?$ and $\B{:}$ of \emph{broadcast}, \emph{receive}, and \emph{discard} actions
  are given by $\B\sharp \coloneqq \{b\sharp \mid b\in\B\}$ for $\sharp\in\{!,?,{:}\}$.
$\Act \coloneqq \Ch \djcup \bar{\Ch} \djcup \{\tau\} \djcup \B! \djcup \B? \djcup \Sig$
  is the set of \emph{actions}, with $\tau$ the \emph{internal action},
  and $\Lab \coloneqq \Act \djcup \B{:} \djcup \bar{\Sig}\!$ is the set of \emph{transition labels}.
Complementation extends to $\Ch \djcup \bar{\Ch} \djcup \Sig \djcup \bar{\Sig}$ by $\bar{\bar{c}} \coloneqq c$.

Below, $c$ ranges over $\Ch \djcup \bar{\Ch} \djcup \Sig \djcup \bar{\Sig}$,
  $\eta$ over $\Ch \djcup \bar{\Ch} \djcup \{\tau\} \djcup \Sig \djcup \bar{\Sig}$,
  $\alpha$ over $\Act$,
  $\ell$ over $\Lab$,
  $\gamma$ over $\In \coloneqq \Lab{\setminus}\Act$,
  $b$ over $\B$,
  $\sharp,\sharp_1,\sharp_2$ over $\{!,?,:\}$, 
  $s$ over $\Sig$, $S$ over recursive specifications and $X$ over $V_S$.
A \emph{relabelling} is a function $f:(\Ch\to\Ch)\djcup(\B\to\B)\djcup(\Sig\to\Sig)$;
  it extends to $\Lab$ by $f(\bar{c})=\overline{f(c)}$, $f(\tau)\coloneqq\tau$ and $f(b\sharp)=f(b)\sharp$.

Next to the constant and operators of CCS,
the process signature $\Sigma$ of \ABCdE features a unary \emph{signalling} operator $\_\!\_\signals s$ for each signal $s \in \Sig$.

\begin{table}[t]
  \vspace{-1.5ex}
  \caption{Structural operational semantics of \ABCdE\label{tab:ABCdE transition rules}}
  \vspace{1ex}
  \centering
  \framebox{$\begin{array}{ccc}
    \displaystyle\frac{}{\mathbf{0} \goto{b{:}} \mathbf{0}}                                   \mylabelA{disNil}{b{:}\mathbf{0}} &
    \displaystyle\frac{\alpha\neq b?}{\alpha.x \goto{b{:}} \alpha.x}                          \mylabelA{disAlpha}{b{:}\alpha.} &
    \displaystyle\frac{x \goto{b{:}} x',~ y \goto{b{:}} y'}{x+y \goto{b{:}} x'+y'}            \mylabelA{plusC}{+^{}_{\!\!\mathrm{\scriptscriptstyle C}}} \\[3ex]
    \multicolumn{3}{c}{
      \displaystyle\frac{x \goto{b\sharp_1} x',~ y \goto{b\sharp_2} y' \quad {\color{Green}(\sharp_1\circ\sharp_2 = \sharp \neq \_)}}
        {x|y \goto{b\sharp} x'|y'}                                                            \mylabelA{parB}{|^{}_\mathrm{\scriptscriptstyle C}}
        \color{Green}\textstyle\qquad~\text{with}~
        \begin{array}{c@{\ }|@{\ }c@{\ \ }c@{\ \ }c}
          \scriptstyle \circ & \scriptstyle ! & \scriptstyle ? & \scriptstyle : \\
          \hline
          \scriptstyle ! & \scriptstyle \_ & \scriptstyle ! & \scriptstyle ! \\
          \scriptstyle ? & \scriptstyle !  & \scriptstyle ? & \scriptstyle ? \\
          \scriptstyle : & \scriptstyle !  & \scriptstyle ? & \scriptstyle : \\
        \end{array}} \\[5.5ex]
    \displaystyle\frac{}{x\signals s \goto{\bar{s}} x\signals s}                              \mylabelA{sigS}{(\sigsyn{s})} &
    \displaystyle\frac{x \goto{\bar{s}} x'}{x+y \goto{\bar{s}} x'+y}                          \mylabelA{plusLE}{+^\mathrm{\!\scriptscriptstyle e}_{\!\!\scriptscriptstyle L}} &
    \displaystyle\frac{y \goto{\bar{s}} y'}{x+y \goto{\bar{s}} x+y'}                          \mylabelA{plusRE}{+^{e}_{\!\!\Right}} \\[4ex]
    \displaystyle\frac{x \goto{\alpha} x'}{x\signals s \goto{\alpha} x'}                      \mylabelA{sigAct}{{}\signals s_\Act} &
    \displaystyle\frac{x \goto{\gamma} x'}{x\signals s \goto{\gamma} x'\signals s}            \mylabelA{sigInd}{{}\signals s_\In} &
    \displaystyle\frac{\rec{S_X|S} \goto{\gamma} y}{\rec{X|S} \goto{\gamma} \rec{X|S}}        ~\myrefA{recIn}
  \end{array}$}
\vspace{-1ex}
\end{table}

The semantics of \ABCdE is given by the transition rule
templates displayed in Tables \ref{tab:CCS transition rules} and \ref{tab:ABCdE transition rules}.
The latter augments CCS with mechanisms for broadcast communication and signalling.
The rule $\myrefA{parB}$ presents the core of broadcast communication~\cite{Prasad91}, where any
  broadcast-action $b!$ performed by a component in a parallel composition needs to synchronise with
  either a receive action $b?$ or a discard action $b{:}$ of any other component.
In order to ensure associativity of the parallel composition, rule $\myrefA{parB}$ also allows
  receipt actions of both components ($\sharp_1 = \sharp_2 = \mathord{?}$),
  or a receipt and a discard, to be combined into a receipt action.

A transition \plat{$p \goto{~b:} q$} is derivable only if $q=p$.
It indicates that the process $p$ is unable to receive a broadcast communication $b!$ on channel $b$.
The Rule $\myrefA{disNil}$ allows the nil process (inaction) to discard any incoming message; 
  in the same spirit $\myrefA{disAlpha}$ allows a message to be discarded by a process that cannot receive it.
A process offering a choice can only perform a discard-action 
if both choice-options can discard the corresponding broadcast (Rule {$\myrefA{plusC}$}). 
Finally, by rule $\myrefA{recIn}$, a recursively defined process $\rec{X|S}$ can discard a broadcast
  iff $\rec{S_X|S}$ can discard it.
The variant $\myrefA{recIn}$ of $\myrefA{recAct}$ is introduced to maintain the property that
  $\target(\theta)=\source(\theta)$ for any indicator transition $\theta$.

A signalling process $p\signals s$ emits the signal $s$ to be read by another process. 
A typically example is a traffic light being red.
Signal emission is modelled as an indicator transition, which does not change the state of the emitting process.
The first rule $\myrefA{sigS}$ models the emission $\bar{s}$ of signal $s$ to the environment. 
The environment (processes running in parallel) can read the signal by performing a read action $s$.
This action synchronises with the emission $\bar{s}$, via rule $\myrefA{parH}$ of \tab{CCS transition rules}.
Reading a signal does not change the state of the emitter.

Rules $\myrefA{plusLE}$ and $\myrefA{plusRE}$ describe the interaction between signal emission and choice.
Relabelling and restriction are handled by rules $\myrefA{restr}$ and $\myrefA{relab}$ of
\tab{CCS transition rules}, respectively.
These operators do not prevent the emission of a signal, and emitting signals never changes the state of the emitting process.
Signal emission $p\signals s$ does not block other transitions of $p$.

It is trivial to check that the TSS of \ABCdE is in De Simone format.

\begin{table}[t]
  \vspace{-1.5ex}
  \caption{Transition signature of \ABCdE}\label{tab:ABCdE transition signature}
  \vspace{1ex}
  \centering
  \renewcommand{\arraystretch}{1.5}
  $\begin{array}{|@{\;}c@{\;}|@{\;}c@{\;}|@{\;}c@{\;}|@{\;}c@{\;}|@{\;}c@{\;}|@{\;}c@{\;}|@{\;}c@{\;}|@{\;}c@{\;}|@{\;}c@{\;}|@{\;}c@{\;}|@{\;}c@{\;}|@{\;}c@{\;}|@{\;}c@{\;}|@{\;}c@{\;}|@{\;}c@{\;}|@{\;}c@{\;}|@{\;}c@{\;}|}
    \hline
    \text{Constructor} &
    \myrefA{actAlpha} & \myrefA{sigS} &
    \myrefA{disNil} & \myrefA{disAlpha} &
    \myrefA{plusL} & \myrefA{plusR} & \myrefA{plusC} & \myrefA{plusLE} & \myrefA{plusRE} &
    \myrefA{parL} & \mylabelA{parCold}{|^{}_\mathrm{\scriptscriptstyle C}} & \myrefA{parR} &
    \myrefA{restr} & \myrefA{relab} &
    \myrefA{sigAct} & \myrefA{sigInd} \\
    \hline
    \text{Arity} &
    1 & 1 &
    0 & 1 &
    2 & 2 & 2 & 2 & 2 &
    2 & 2 & 2 &
    1 & 1 &
    1 & 1 \\
    \hline
    \text{Trigger Set} &
    \emptyset & \emptyset &
    \emptyset & \emptyset &
    \{1\} & \{2\} &
    \{1{,}2\} & \{1\} & \{2\} & \{1\} & \{1{,}2\} &
    \{2\} & \{1\} & \{1\} &
    \{1\} & \{1\} \\
    \hline
  \end{array}$
\end{table}

The transition signature of \ABCdE (\tab{ABCdE transition signature}) is completely determined by the set of
transition rule templates in Tables \ref{tab:CCS transition rules} and \ref{tab:ABCdE transition rules}.
We have united the rules for handshaking and broadcast communication
  by assigning the same name $\myrefA{parCold}$ to all their instances. 
  When expressing transitions in ABCdE as expressions,
   we use infix notation for the binary transition constructors,
  and prefix or postfix notation for unary ones.
For example, the transition $\myrefA{disNil}()$ is shortened to $\myrefA{disNil}$,
  $\myrefA{actAlpha}(p)$ to $\myrefA{actAlpha}p$,
  $\myrefA{restr}(t)$ to $t\myrefA{restr}$,
  and $\myrefA{parL}(t,p)$ to $t {\myrefA{parL}} p$.

\begin{wrapfigure}[12]{r}{0.34\textwidth}
  \vspace{-2.5ex}
  
  \renewcommand{\arraystretch}{0.95}
  $\begin{array}{|c|c|}
    \hline
    \text{Meta} & \text{Variable Expression} \\
    \hline
    P & x_1 \\
    Q & x_2 \\
    P' & y'_1 \\
    Q' & y'_2 \\
    t & (\tx_1 \dblcolon x_1 \goto{\mathit{xa}_1} x'_1) \\
    u & (\tx_2 \dblcolon x_2 \goto{\mathit{xa}_2} x'_2) \\
    v & (\ty_1 \dblcolon x_1 \goto{\mathit{ya}_1} y'_1) \\
    w & (\ty_2 \dblcolon x_2 \goto{\mathit{ya}_2} y'_2) \\
    t' & (\tz_1 \dblcolon y'_1 \goto{\mathit{za}_1} z'_1) \\
    u' & (\tz_2 \dblcolon y'_2 \goto{\mathit{za}_2} z'_2) \\
    \hline
  \end{array}$
\end{wrapfigure}

\mbox{}
\tab{ABCdE successor rules} extends the successor relation of CCS (\tab{CCS successor rules}) to ABCdE. 
$P,Q$ are process variables,
  $t,v$ transition variables enabled at $P$,
  $u,w$ transition variables enabled at $Q$,
  $P',Q'$ the targets of $v,w$, respectively and
  $t',u'$ transitions enabled at $P',Q'$, respectively.
To express those rules in the same way as \df{TSSS in De Simone format},
  we replace the metavariables $P$, $Q$, $t$, $u$, etc.\ with variable expressions as indicated on the right.
  Here $\mathit{xa}_i$, $\mathit{ya}_i$, $\mathit{za}_i$ are label variables that should be instantiated to match
  the trigger of the rules and side conditions.
  As \ABCdE does not feature operators of arity~${>}2$, the index~$i$ from \df{TSSS in De Simone format}
  can be 1 or 2 only.
  
\begin{table}[bt]
  \vspace{-1.5ex}
  \caption{Successor rules for \ABCdE\label{tab:ABCdE successor rules}}
  \vspace{1ex}
  \centering
  \framebox{$\begin{array}{@{}c@{\ \ \ \ }c@{\ \ \ \ }c@{\ \ \ \ }c@{}}
    \multicolumn{4}{c}{
      \displaystyle\frac{{\color{Green} \ell(t)\in\{b?,b{:}\}}}
        {\hyperlink{lab:actAlpha}{{\color{blue} \actsyn{b?}}}P \leadsto_{\hyperlink{lab:actAlpha}{{\color{blue} \actsyn{b?}}}P} t}\mylabelS{2a}^b \qquad
      \displaystyle\frac{{\color{Green} \ell(\zeta)\in\B{:}\cup\bar{\Sig}}}{\chi \leadsto_\zeta \chi}\mylabelS{1} \qquad
      \displaystyle\frac{{\color{Green} \ell(t)\in\{b?,b{:}\}}}{{\myrefA{disAlpha}}P \leadsto_{{\myrefA{actAlpha}}P} t}\mylabelS{2b}^b} \\[4ex]
    \displaystyle\frac{t \leadsto_v t'}
      {\begin{array}{@{}c@{}}
        t {\myrefA{plusLE}} Q \leadsto_{v {\myrefA{plusL}} Q} t'
      \end{array}}\mylabelS{3a } &
    \displaystyle\frac{t \leadsto_v t'}{t {\myrefA{plusC}} u \leadsto_{v {\myrefA{plusL}} Q} t'}\mylabelS{3b} &
    \displaystyle\frac{u \leadsto_w u'}
      {\begin{array}{@{}c@{}}
        P {\myrefA{plusRE}} u \leadsto_{P {\myrefA{plusR}} w} u'
        \end{array}}\mylabelS{4a } &
    \displaystyle\frac{u \leadsto_w u'}{t {\myrefA{plusC}} u \leadsto_{P {\myrefA{plusR}} w} u'}\mylabelS{4b} \\[4ex]

    \multicolumn{2}{c}{
      \displaystyle\frac{{\color{Green} \ell(t)=b? \quad \ell(u')\in\{b?,b{:}\}}}
        {\begin{array}{@{}c@{}}
          t {\myrefA{plusL}} Q \leadsto_{P {\myrefA{plusR}} w} u' \\
          t {\myrefA{plusLE}} Q \leadsto_{P {\myrefA{plusR}} w} u'
        \end{array}}\mylabelS{5}^b} &
    \multicolumn{2}{c}{
      \displaystyle\frac{{\color{Green} \ell(u)=b? \quad \ell(t')\in\{b?,b{:}\}}}
        {\begin{array}{@{}c@{}}
          P {\myrefA{plusR}} u \leadsto_{v {\myrefA{plusL}} Q} t' \\
          P {\myrefA{plusRE}} u \leadsto_{v {\myrefA{plusL}} Q} t'
        \end{array}}~\mylabelS{6}^b} \\[7ex]
    
    \multicolumn{2}{c}{
      \displaystyle\frac{t \leadsto_v t'}
        {\begin{array}{@{}c@{}}
          \RecIn(X,S,t) \leadsto_{\RecAct(X,S,v)} t'
        \end{array}}\mylabelS{11c }} &
    \multicolumn{2}{c}{
      \displaystyle\frac{t \leadsto_v t'}
        {\begin{array}{@{}c@{}}
          t{\myrefA{sigAct}} \leadsto_{v{\myrefA{sigAct}}} t' \\
          t{\myrefA{sigInd}} \leadsto_{v{\myrefA{sigAct}}} t'
        \end{array}}\mylabelS{11d}}
  \end{array}$}\vspace{-1pt}
\end{table}

To save duplication of rules~\myrefS{8b}, \myrefS{8c}, \myrefS{9b}, \myrefS{9c} and~\myrefS{10} 
we have assigned the same name $\myrefA{parCold}$ to the rules for
handshaking and broadcast communication.
The intuition of the rules of \tab{ABCdE successor rules} is explained in detail in \cite{GHW21ea}. 

In the naming convention for transitions from \cite{GHW21ea} the sub- and superscripts of the transition constructors
  ${\color{blue} +}$, ${\color{blue} |}$ and ${\color{blue} {}\signals s}$, and of the recursion construct, were suppressed.
In most cases that yields no ambiguity, as the difference between $\myrefA{parL}$ and $\myrefA{parR}$, for instance,
  can be detected by checking which of its two arguments are of type transition versus process.
Moreover, it avoids the duplication in rules~\myrefS{3a}, \myrefS{4a}, \myrefS{5}, \myrefS{6}, \myrefS{11c} and~\myrefS{11d}.
The ambiguity between $\myrefA{plusL}$ and $\myrefA{plusLE}$ (or $\myrefA{plusR}$ and $\myrefA{plusRE}$)
  was in \cite{GHW21ea} resolved by adorning rules~\hyperlink{lab:3a}{{\color{orange} 3}}--\myrefS{6} with a side condition $\ell(v)\notin\bar{\Sig}$
  or $\ell(w)\notin\bar{\Sig}$, and the ambiguity between $\myrefA{recAct}$ and $\myrefA{recIn}$ (or $\sigAct$ and $\sigInd$)
  by adorning rules~\myrefS{11c} and~\myrefS{11d} with a side condition $\ell(v)\in\Act$; this is not needed here.

It is easy to check that all rules are in the newly introduced De Simone format, 
except Rule~\myrefS{1}.
However, this rule can be converted in to a collection of De Simone rules by substituting $\snr(\mathit{xe}_1,\dots,\mathit{xe}_n)$ for $\chi$
  and $\sns(\mathit{ye}_1,\dots,\mathit{ye}_n)$ for $\zeta$, adding a premise in the form of $\mathit{xe}_i \leadsto_{\mathit{ye}_i} (\tz_i \dblcolon y'_i \goto{\mathit{za}_i} z'_i))$ if $i\in I_\snr \cap I_\sns$, for each pair of rules of the same type named $\snr$ and $\sns$.%
\footnote{This yields $1^2 + 2 \cdot 1 + 5 \cdot 3 + 3 \cdot 2 + 2 \cdot 1 = 26$ rules of types $(\mathbf{0},0)$, $(\alpha.\_\!\_,1)$,
  $(+,2)$, $({}\signals s,1)$ and $\rec{X|S}$ not included in Tables~\ref{tab:CCS successor rules} and \ref{tab:ABCdE successor rules}.}
The various occurrences of \myrefS{1} in \fig{inclusions} refer to these substitution instances.
It follows that $\bisep$ is a congruence for the operators of \ABCdE, as well as a lean congruence for recursion.

\section{Related Work \& Conclusion}

In this paper we have added a successor relation to the well-known De Simone format.
This has allowed us to prove the general result that enabling preserving 
bisimilarity -- a finer semantic equivalence relation than strong 
bisimulation -- is a lean congruence for all languages with a structural operational semantics
within this format. We do not cover full congruence yet, as proofs for general
recursions are incredible hard and usually excluded from work justifying semantic equivalences.

There is ample work on congruence formats in the literature. 
Good overview papers are~\cite{AFV00,MRG07}. 
For system description languages that do not capture time, probability or other useful extensions to
standard process algebras, all congruence formats target strong bisimilarity, or some semantic
equivalence or preorder that is strictly coarser than strong bisimilarity. As far as we know, the
present paper is the first to define a congruence format for a semantic equivalence that is finer than strong bisimilarity.

Our congruence format also ensures a lean congruence for recursion. So far,
the only papers that provide a rule format yielding a congruence property for recursion
are \cite{Re00} and \cite{vG17b}, and both of them target strong bisimilarity.

In Sections \ref{sec:CCS} and \ref{sec:abcde}, we have applied our format to show lean congruence of ep-bisimilarity for the process
algebra CCS and \ABCdE, respectively.  This latter process algebra features broadcast communication and signalling.
These two features are representative for issues that may arise elsewhere, and help to
ensure that our results are as general as possible. Our congruence format can
effortlessly be applied to other calculi like CSP~\cite{BHR84} or ACP~\cite{BW90}.

In order to evaluate ep-bisimilarity on process algebras like CCS, CSP, ACP or \ABCdE,
their semantics needs to be given in terms of labelled transition systems extended with a successor
relation $\leadsto$.
This relation models concurrency between transitions enabled in the same state, and also
tells what happens to a transition if a concurrent transition is executed first.
Without this extra component, labelled transition systems lack the necessary information to capture
liveness properties in the sense explained in the introduction.

In a previous paper \cite{GHW21ea} we already gave such a semantics to \ABCdE.
The rules for the successor relation presented in \cite{GHW21ea}, displayed in
Tables \ref{tab:CCS successor rules} and \ref{tab:ABCdE successor rules}, are now
seen to fit our congruence format. We can now also conclude that ep-bisimulation is a lean congruence for \ABCdE.
In \cite[Appendix~B]{GHW21} we contemplate a very different approach for defining the relation $\leadsto$.
Following \cite{vG19}, we understand each transition as the synchronisation of a number of elementary particles called \emph{synchrons}.
Then relations on synchrons are proposed in terms of which the $\leadsto$-relation is defined.
It is shown that this leads to the same $\leadsto$-relation as the operational approach from \cite{GHW21ea} and
Tables \ref{tab:CCS successor rules} and \ref{tab:ABCdE successor rules}.

\newpage
\bibliography{refs}

\newpage
\appendix

\section{Proof of Lean Congruence}\label{app:a congruence result}

Before presenting the proof of lean congruence (\thm{lean congruence}) we 
present some concepts needed later on: we first introduce consistency of substitutions;
then we  define \emph{standard} open transitions and relate closed substitutions.

Substitutions are partial functions.
Often only a subset of the substitution function affects the substitution instance.
To compare two substitutions with regard to a given domain, we define function consistency: 
two (partial) functions $f,g$ with domains $D_f$ and $D_g$, respectively, are \emph{consistent}\index{consistent} if for all $\chi\in D_f\cap D_g$, $f(\chi)=g(\chi)$.
  $f,g$ are consistent on $D\subseteq D_f\cap D_g$ if for all $\chi\in D$, $f(\chi)=g(\chi)$.

\begin{definitionc}{Standard Open Transition}\label{df:standard open transition}\upshape
  An open transition $\mathring{t}$ is \emph{standard}\index{standard} if there exists an injective \emph{target variable function}\index{target variable function}
    $\Vtar[\mathring{t}]\!:\Tvar(\mathring{t})\mathrel\to\Var\setminus\var(\Osrc(\mathring{t}))$ such that
    for all $\tx\mathbin\in\Tvar(\mathring{t})$,\\ \plat{$\beta_{\mathring{t}}(\tx)=(x \goto{a} \Vtar(\tx))$} for some $x\in\Var$.
  $\Vtar[\mathring{t}](\tx)$ is then called the \emph{target variable} of $\tx$ in $\mathring{t}$.

  The set $\var(\mathring{t})$ of \emph{free process variables} of a standard open transition $\mathring{t}$
  contains those that occur either in $\var(\Osrc(\mathring{t}))$ or
  as the target variables of some transition variable $\tx$ in $\Tvar(\mathring{t})$.
\end{definitionc}

\begin{definitionc}{Relating Closed Substitutions}\label{df:relating closed substitutions}\upshape
  Let \tss be a TSS and $C\!:\Var\to\Pow(\Tr\times\Tr)$.
  Two closed $\Sigma$-substitutions $\rho,\nu$ are $C$-ep-bisimilar on $p\in\IP^r(\Sigma)$ if, 
    for all $x\in\var(p)$, $(\rho(x),\nu(x),C(x))\in\R$ for some ep-bisimulation $\R$.

  Fix a TSSS \tsss. 
  Two closed $\itt$-substitutions $\Trho,\Tnu$ are $C$-ep-bisimilar on a standard open transition $\mathring{t}$ if
  \leftmargini20pt
    \begin{itemize}
      \item
       $\Trho{\upharpoonright}\Var,\Tnu{\upharpoonright}\Var$
        are $C$-ep-bisimilar on all process expressions in $\mathring{t}$, and
      \item for all $\tx\mathbin\in\Tvar(\mathring{t})$, if $\beta_{\mathring{t}}(\tx)\mathbin=(x \goto{a} y)$ then 
        \leftmarginii12pt
          \begin{itemize}
                \item $\ell(\Trho(\tx))\mathbin=\ell(\Tnu(\tx))\mathbin=a$, $\Trho(y)\mathbin=\target(\Trho(\tx))$, $\Tnu(y)\mathbin=\target(\Tnu(\tx))$, $(\Trho(\tx),\Tnu(\tx))\mathbin\in C(x)$,             
                \item if $(t,u)\in C(x)$ and $t\leadsto_{\Trho(\tx)}t'$
                        then $\exists\, u'.~u\leadsto_{\Tnu(\tx)} u' \land (t',u')\in C(y)$ and 
                \item if $(t,u)\in C(x)$ and $u\leadsto_{\Tnu(\tx)}u'$
                        then $\exists\, t'.~ t\leadsto_{\Trho(\tx)}t' \land (t',u')\in C(y)$.
              \end{itemize}
    \end{itemize}
\end{definitionc}

\begin{observation}\label{obs:C existence}\upshape
  Fix a TSS \tss.
  If $\rho,\nu$ are two closed $\Sigma$-substitutions such that $\forall x\in\Var.\linebreak \rho(x) \bisep \nu(x)$, 
    then, for all $p\in\IP^r(\Sigma)$, $\rho,\nu$ are $C$-ep-bisimilar on $p$ for some $C:\Var\to\Pow(\Tr\times\Tr)$.
\end{observation}

\begin{observation}\label{obs:transition bisimilarity property}\upshape
  Fix a TSSS $\tsss$.
  If $\mathring{t}$ is a standard open transition and $\Trho,\Tnu$ are two closed $\itt$-substitutions
    which are $C$-ep-bisimilar on $\mathring{t}$ for some $C\!:\Var\to\Pow(\Tr\times\Tr)$ then $\Trho,\Tnu$ match $\mathring{t}$.
\end{observation}

\noindent
The next lemma allows us to rename transition variables and target variables in standard open transitions.

\begin{lemma}\label{lem:renaming open transition}\upshape
  Fix a TSSS $\tsss$.
  Suppose $p\in\IP^r(\Sigma)$ is a process expression, $\mathring{t}\in\Oen(p)$ is a standard open transition, $C\!:\Var\to\Pow(\Tr\times\Tr)$ and
    $\Trho,\Tnu$ are two closed $\itt$-substitutions which are $C$-ep-bisimilar on $\mathring{t}$.
  If $W \subseteq \Var\setminus\var(p)$ and $W^\TT \subseteq \TVar$ are two infinite sets of variables, then there exist
    a standard open transition $\mathring{t}'\in\Oen(p)$, two closed $\itt$-substitutions $\Trho',\Tnu'$ and $C'\!:\Var\to\Pow(\Tr\times\Tr)$
    such that $\Tvar(\mathring{t}') \subseteq W^\TT$, the image of $\Vtar[\mathring{t}']$ is a subset of $W$,
    $\mathring{t}'[\Trho']=\mathring{t}[\Trho]$, $\mathring{t}'[\Tnu']=\mathring{t}[\Tnu]$, $C'$ is consistent with $C$ on $\var(p)$,
    $\Trho',\Tnu'$ are consistent with $\Trho,\Tnu$ respectively on $\var(p)$ and are $C'$-ep-bisimilar on $\mathring{t}'$.
\end{lemma}
\begin{proof}
  Let $t := \mathring{t}[\Trho]$, $u := \mathring{t}[\Tnu]$, $\rho := \Trho{\upharpoonright}\Var$ and $\nu := \Trho{\upharpoonright}\Var$.
  We prove the lemma by induction on the size of $\mathring{t}$, applying a case distinction on the shape of $p$.
  \begin{itemize}
    \item If $p=x\in\Var$ then $\mathring{t}$ must be of the form \plat{$(\tx \dblcolon x \goto{a} y)$} where $a=\Oell(\mathring{t})$ and $y\in\Var\setminus\{x\}$.
          Pick $\tx'\in W^\TT$ and $z\in W$ and let $\mathring{t}'$ be $(\tx' \dblcolon x \goto{a} z)$.
          We can update $\rho$ to $\Trho'$ by $\Trho'[\tx']=t$ and $\Trho'(z)=\target(t)$.
          $\Tnu'$ is retrieved similarly.
          Let $C'(x)=C(x)$ and $C'(z)=C(y)$.
    \item If $p$ is of the form $\Op(p_1,\dots,p_n)$ then $\mathring{t}$'s expression must be of the form $\snr(E_1,\dots,E_n)$.
          Let $I$ be the trigger set of rules named $\snr$.
          Let $\{W_i \mid i\in I\}$ be a partition of $W$ such that each $W_i$ is infinite.
          Find $W^\TT_i$ for $i\in I$ similarly.
          By induction hypothesis, for each $i\in I$ there exist a standard open transition $\mathring{t}'_i$, two closed $\itt$-substitutions $\Trhoi',\Tnui'$ and $C'_i\!:\Var\to\Pow(\Tr\times\Tr)$
            such that $\Tvar(\mathring{t}'_i) \subseteq W_i^\TT$, the image of $\Vtar[\mathring{t}'_i]$ is a subset of $W_i$,
            $\mathring{t}'_i[\Trhoi']=E_i[\Trho]$, $\mathring{t}'_i[\Tnui']=E_i[\Tnu]$, $C'_i$ is consistent with $C$ on $\var(\Osrc(E_i))$,
            $\Trhoi',\Tnui'$ are consistent with $\Trho,\Tnu$ respectively on $\var(\Osrc(E_i))$ and are $C'_i$-ep-bisimilar on $\mathring{t}'_i$.
          Let $\mathring{t}'={\snr}(E'_1,\dots,E'_n)$, with $E'_i=\mathring{t}'_i$ for $i\in I$ and $E'_i=p_i$ otherwise.
          Again, we update $\rho$ to $\Trho'$ by $\Trho'(\tx)=\Trhoi'(\tx)$ and $\Trho'(z)=\target(\Trhoi'(\tx))$ if $\tx$ occurs in $\mathring{t}'_i$ and its target variable is $z$.
          $\Tnu'$ is retrieved similarly.
          Let $C'(x)=C(x)$ for $x\in\var(p)$ and $C'(z)=C'_i(z)$ if $z$ occurs in $\mathring{t}'_i$.
    \item The case when $p=\rec{X|S}$ is handled in the same way as above:
            $\mathring{t}$'s expression must be of the form $\RecAct(X,S,E)$ or $\RecIn(X,S,E)$ where $E$ is the direct subtransition of $\mathring{t}$ enabled at $\rec{S_X|S}$.
  \end{itemize}
  For each case, it is straightforward to check that $\mathring{t}',\Trho',\Tnu',C'$ satisfy the requirement.
\end{proof}
When we build a standard open transitions out of its direct subtransitions, it is essential to
ensure the validity of the target variable function -- cf. \df{standard open transition}.
\lem{renaming open transition} allows us to assume that the relevant subtransitions do not share any transition variables or target variables.
Thus their target variable functions will have disjoint domains.
The union of these functions will be injective, and thus a valid target variable function for the composed standard open transition.

\vspace{2mm}
\newenvironment{proofThm}[1]
  {\noindent \textbf{Proof of #1:\ }}
  {}
\begin{proofThm}{\thm{lean congruence}}
  Let $\R\subseteq S\times S\times\Pow(\Tr\times\Tr)$ be the smallest relation satisfying
    \begin{itemize}
      \item if $(p,q,R)\in\R'$ for some ep-bisimulation $\R'$ then $(p,q,R)\in\R$, and
      \item if $p\in\IP^r(\Sigma)$ and $\rho,\nu$ are two closed $\Sigma$-substitutions which are $C$-ep-bisimilar on $p$
              for some $C\!:\Var\to\Pow(\Tr\times\Tr)$, then
              $(p[\rho],p[\nu],R(p,\rho,\nu,C))\mathbin\in\R$, where
                \[R(p,\rho,\nu,C) := \left\{(\mathring{t}[\Trho],\mathring{t}[\Tnu])\left|\,
                \parbox{2.5in}{
                  $\mathring{t}\in\Oen(p)$ is standard and\\
                  $\Trho,\Tnu$ are closed
                    $\itt$-substitutions\\ that are consistent with $\rho,\nu$
                    respectively on $\var(p)$ and are $C'$-ep-bisimilar
                    on $\mathring{t}$\\ for some $C':\Var\to\Pow(\Tr\times\Tr)$\\ that is
                    consistent with $C$ on $\var(p)$
                }
                 \right.\right\}.
              \]
    \end{itemize}
  By \obs{C existence}, it suffices to show that $\R$ is an ep-bisimulation.
  That means, each tuple in $\R$ satisfies the requirements in \df{ep-bisimilarity}.

  If the tuple stems from the first clause then it must have satisfied all the requirements.

  If it stems from the second clause, it has the form $(p[\rho],p[\nu],R(p,\rho,\nu,C))$ for certain $p,\rho,\nu,C$.
     Fix $\rho,\nu,C$ and write $R(p,\rho,\nu,C)$ as $R(p)$.
  
  \vspace{1.5ex}
  \noindent
  \textbf{Part 1.}
  For each $(\mathring{t}[\Trho],\mathring{t}[\Tnu])\in R(p)$,
    \obs{transition bisimilarity property} tells us that $\Trho,\Tnu$ match $\mathring{t}$;
    \obs{transition composition} further confirms the validity of the substitution instances and informs us of their sources,
    i.e., $\source(\mathring{t}[\Trho\!])\mathbin=\Osrc(\mathring{t})[\Trho]=p[\Trho]=p[\rho]$ and similarly $\source(\mathring{t}[\Tnu])=p[\nu]$.
  Thus $R(p)\subseteq\en(p[\rho])\times\en(p[\nu])$.

  By symmetry, we (only) need to show that $\forall t\in\en(p[\rho]).~ \exists\, u.~ (t,u)\in R(p) \land \ell(t)=\ell(u)$.
  Instead of proving it for a specific $p$, below we prove it for all open process expressions $p$
  such that $\rho,\nu$ are $C$-ep-bisimilar on $p$.

  It suffices to find, for each pair $(p,t)$, a standard open transition $\mathring{t}\in\Oen(p)$, two closed $\itt$-substitutions $\Trho,\Tnu$
    and $C':\Var\to\Pow(\Tr\times\Tr)$, such that $\mathring{t}[\Trho]=t$, $C'$ is consistent with $C$ on $\var(p)$,
    $\Trho,\Tnu$ are consistent with $\rho,\nu$ respectively on $\var(p)$ and are $C'$-ep-bisimilar on $\mathring{t}$.
  We proceed by induction on the size of $t$, applying a case distinction on the shape of $p$.
  \begin{itemize}
    \item If $p=x\in\Var$ then $\mathring{t}$ must be of the form \plat{$(\tx \dblcolon x \goto{a} y)$} where $a=\ell(t)$ and $y\in\Var\setminus\{x\}$.
          We can update $\rho$ to $\Trho$ by $\Trho(\tx)=t$ and $\Trho(y)=\target(t)$.
          $\Tnu$ can be found by letting $\Tnu(\tx)=u$ and $\Tnu(y)=\target(u)$ for some $u$ with $(t,u)\in C(x)$.
          Let $C'(x)=C(x)$ and $C'(y)=R_y$ for some $R_y$ with $(\target(t),\target(u),R_y)$ in some ep-bisimulation and
            \begin{itemize}
              \item if $(v,w)\in C(x)$ and $v \leadsto_t v'$ then $\exists\, w'.~ w \leadsto_u w' \land (v',w')\in R_y$, and
              \item if $(v,w)\in C(x)$ and $w \leadsto_u w'$ then $\exists\, v'.~ v \leadsto_t v' \land (v',w')\in R_y$.
            \end{itemize}
          By \df{relating closed substitutions}, we know that $(\rho(x),\nu(x),C(x))$ is a member of some ep-bisimulation, and thus the existence of $R_y$ is guaranteed by \df{ep-bisimilarity}.
    \item If $p$ is of the form $\Op(p_1,\dots,p_n)$ then $t$'s expression must be of the form $\snr(E_1,\dots,E_n)$.
          Let $I$ be the trigger set of rules named $\snr$.
          For each $i\in I$, by \df{transition expression}, $E_i$ is a direct subtransition of $t$, which is a transition enabled at $p_i[\rho]$.
          Then by the induction hypothesis, we can find the corresponding $\mathring{t}_i,\Trhoi,\Tnui,C'_i$.
          By \lem{renaming open transition}, we assume that the $\mathring{t}_i$ for $i\in I$ only use process variables in $\Var\setminus\var(p)$ as target variables,
            and do not share any common transition variables and target variables.
          This ensures that $\mathring{t}=\snr(E'_1,\dots,E'_n)$, with $E'_i=\mathring{t}_i$ for $i\in I$ and $E'_i=p_i$ otherwise, is a standard open transition.
          We update $\rho$ to $\Trho$ by $\Trho(\tx)=\Trhoi(\tx)$ and $\Trho(y)=\target(\Trhoi(\tx))$ if $\tx$ occurs in $\mathring{t}_i$ and its target variable is $y$.
          $\Tnu$ is retrieved similarly.
          Let $C'(x)=C(x)$ for $x\in\var(p)$ and $C'(y)=C'_i(y)$ if $y$ occurs in $\mathring{t}_i$.
    \item The case when $p=\rec{X|S}$ is handled in the same way as above:
            $t$'s expression must be of the form $\RecAct(X,S[\rho\setminus V_S],E)$ or $\RecIn(X,S[\rho\setminus V_S],E)$, where $E$ is the direct subtransition of $t$ enabled at $\rec{S_X|S}[\rho]$.
  \end{itemize}

  \newcommand{\W}{\altmathcal{W}} 
  \noindent
  \textbf{Part 2.}
  Recall that we have fixed $\rho,\nu,C$ and $R(p)$ denotes $R(p,\rho,\nu,C)$.
  We will show that, for all $p$, $v$ and $w$, if $(v,w)\in R(p)$ then there exist a process expression $p^\W$, two closed $\Sigma$-substitutions $\rho^\W,\nu^\W$ and $C^\W:\Var\to\Pow(\Tr\times\Tr)$, such that
    \begin{enumerate}[label=\textbf{(\alph*)}, left=0pt]
      \item $\rho^\W$ and $\nu^\W$ are $C^\W$-ep-bisimilar on $p^\W$,
      \item $p^\W[\rho^\W]=\target(v)$, $p^\W[\nu^\W]=\target(w)$,
      \item if $(t,u)\in R(p)$ and $t \leadsto_v t'$, then $\exists\, u.~ u \leadsto_w u' \land (t',u')\in R(p^\W,\rho^\W,\nu^\W,C^\W)$, and
      \item if $(t,u)\in R(p)$ and $u \leadsto_w u'$, then $\exists\, t.~ t \leadsto_v t' \land (t',u')\in R(p^\W,\rho^\W,\nu^\W,C^\W)$.
    \end{enumerate}
  Note that $(p^\W[\rho^\W],p^\W[\nu^\W],R(p^\W,\rho^\W,\nu^\W,C^\W))\in\R$ by definition.
  Hence this statement shows that the tuple $(p[\rho],p[\nu],R(p))$ also satisfies Part 2 of \df{ep-bisimilarity}.

  For $(v,w)\in R(p)$, define a $(p,v,w)$-\emph{witness} to be a tuple $\W := (\mathring{v},\Trho^v,\Tnu^v,C^v)$
    with $\mathring{v}\in\Oen(p)$ standard, $C^v:\Var\to\Pow(\Tr\times\Tr)$ consistent with $C$ on $\var(p)$,
    and $\Trho^v,\Tnu^v$ closed $\itt$-substitutions that are consistent with $\rho,\nu$ respectively on $\var(p)$ and are $C^v$-ep-bisimilar on $\mathring{v}$,
    such that $\mathring{v}[\Trho^v]=v$ and $\mathring{v}[\Tnu^v]=w$.
    By definition, such a witness always exists.
  Define the \emph{successor tuple} of a $(p,v,w)$-witness $\W = (\mathring{v},\Trho^v,\Tnu^v,C^v)$ to be $(p^\W,\rho^\W,\nu^\W,C^\W)$ with $p^\W := \Otar(\mathring{v})$, $\rho^\W := \Trho^v{\upharpoonright}\Var$, $\nu^\W := \Tnu^v{\upharpoonright}\Var$ and $C^\W := C^v$.
  It suffices to show that for every $(p,v,w)$-witness $\W$, its successor tuple $(p^\W,\rho^\W,\nu^\W,C^\W)$ satisfies the requirements \textbf{(a)}--\textbf{(d)}~above.
  
  Let $\W$ be an $(p,v,w)$-witness for some $(v,w)\in R(p)$.
  By \df{relating closed substitutions}, $\rho^\W$ and $\nu^\W$ are $C^\W$-ep-bisimilar on $p^\W$.
  Applying \obs{transition composition}, $\target(v) = \target(\mathring{v}[\Trho^v]) = \Otar(\mathring{v})[\Trho^v] = p^\W[\Trho^v] = p^\W[\rho^\W]$.
  Likewise $\target(w)=p^\W[\nu^\W]$.
  It remains to establish \textbf{(c)}, as \textbf{(d)} then follows by symmetry.
  We prove by induction on the size of $v$ that for each $(p,v,w)$-witness its successor tuple satisfies \textbf{(c)}.
  In this proof we apply a case distinction on the shape of $p$.
  
  Suppose $(t,u)\in R(p)$ and $t \leadsto_v t'$.
  \textit{It suffices to find a standard open transition $\mathring{t}'\in\Oen(p^\W)$, two closed $\itt$-substitutions $\Trho',\Tnu'$ and $C':\Var\to\Pow(\Tr\times\Tr)$, such that
    $\mathring{t}'[\Trho']=t'$, $u \leadsto_w \mathring{t}'[\Tnu']$, $C'$ is consistent with $C^\W$ on $\var(p^\W)$,
    $\Trho',\Tnu'$ are consistent with $\rho^\W,\nu^\W$ respectively on $\var(p^\W)$ and are $C'$-ep-bisimilar on $\mathring{t}'$.}
  The required transition $u'$ will then be $\mathring{t}'[\Tnu']$.
  \begin{itemize}
    \item If $p=x\in\Var$ then $\mathring{v}$ must be of the form \plat{$\tx \dblcolon x \goto{a} y$} where $y=p^\W\in\Var\setminus\{x\}$.
          Since $\mathring{v}[\Trho^v]=v$ and $\mathring{v}[\Tnu^v]=w$, we have $\Trho^v(\tx)=v$ and $\Tnu^v(\tx)=w$,
            so the $C^v$-ep-bisimilarity between $\Trho^v$ and $\Tnu^v$ on $\mathring{v}$ implies $(v,w)\in C^v(x)$ by \df{relating closed substitutions}.
          In the same way one derives $(t,u)\in C^{\,t}(x)$ for some $C^{\,t}:\Var\to\Pow(\Tr\times\Tr)$ that is consistent with $C$ on $\var(p)$.
          As $C^{\,t},C^v$ is consistent with $C$ on $\var(p)=\{x\}$, we have $C^{\,t}(x)=C(x)=C^v(x)$ and thus $(t,u)\in C^v(x)$.
          Thus, by \df{relating closed substitutions}, there exists a $u'$ with $u \leadsto_w u'$ and $(t',u')\in C^v(y) = C^\W(y)$.
          Let \plat{$\mathring{t}'=(\ty \dblcolon y \goto{b} z)$} where $b=\ell(t')$ and $z\in\Var\setminus\{y\}$.
          We can update $\rho^\W$ to $\Trho'$ by $\Trho'(\ty)=t'$ and $\Trho'(z)=\target(t')$.
          $\Tnu'$ is retrieved similarly.
          $C'$ can be found in the same way as in Part 1 for the case when $p=x\in\Var$.
    \item If $p$ is of the form $\Op(p_1,\dots,p_n)$ then $\mathring{v}$'s expression must be of the form $\sns(F_1,\dots,F_n)$.
          Let $I_\sns$ be the trigger set of rules named $\sns$.
          For each $i\in I_\sns$, by \df{transition expression}, $F_i$ is a standard open transition $\mathring{v}_i\in\Oen(p_i)$.
          Therefore, by definition, $(v_i,w_i) := (\mathring{v}_i[\Trho^v],\mathring{v}_i[\Tnu^v])\in R(p_i)$
            and $\W_i := (\mathring{v}_i,\Trho^v,\Tnu^v,C^v)$ is an $(p_i,v_i,w_i)$-witness.
          Thus, by induction, the successor tuple $(p^\W_i,\rho^\W,\nu^\W,C^\W)$ of $\W_i$, where $p^\W_i := \Otar(\mathring{v}_i)$, satisfies Requirement \textbf(c).

          Using that $(t,u)\in R(p)$, by definition, we obtain a $(p,t,u)$-witness $(\mathring{t},\Trho^t,\Tnu^t,C^{\,t})$.
          Applying \lem{renaming open transition}, we may assume that $\mathring{t}$ does not employ
          transition variables from $\Tvar(\mathring{v})$ or target variables that occur in $\var(\mathring{v})$.
          Hence $\var(\mathring{t}) \cap \var(\mathring{v}) = \var(p)$.
          Since $\mathring{t}\in\Oen(p)$, $\mathring{t}$'s expression must be of the form $\snr(E_1,\dots,E_n)$.
          Let $I_\snr$ be the trigger set of rules named $\snr$.
          By \df{transition expression}, for each $i\in I_\snr$, $E_i$ is a standard open transition $\mathring{t}_i\in\Oen(p_i)$ and $E_i=p_i$ otherwise.
          Therefore, for each $i\in I_\snr$, by definition, $(t_i,u_i) := (\mathring{t}_i[\Trho^t],\mathring{t}_i[\Tnu^t])\in R(p_i)$.

          The root of the proof $\pi$ of the successor literal $t \leadsto_v t'$ must apply a substitution instance of a successor rule
          \begin{equation}\label{rootRule}
              \frac{\{(\tx_i \dblcolon x_i \goto{\mathit{xa}_i} x'_i) \leadsto_{(\ty_i \dblcolon x_i \goto{\mathit{ya}_i} y'_i)} (\tz_i \dblcolon y'_i \goto{\mathit{za}_i} z'_i) \mid i\in I\}}
                  {\snr(\mathit{xe}_1,\dots,\mathit{xe}_n) \leadsto_{\sns(\mathit{ye}_1,\dots,\mathit{ye}_n)} \mathring{r}}
          \end{equation}
          in (our enrichment of the) De Simone format.
          Here $I \subseteq I_\snr \cap I_\sns$ by \df{TSSS in De Simone format}.
          This substitution instance must have the form
          \begin{equation}\label{instance}
              \frac{\{t_i \leadsto_{v_i} t'_i \mid i\in I\}}{t \leadsto_v t'} \;.
          \end{equation}
          Therefore $t_i \leadsto_{v_i} t'_i$ for $i\in I$.
          Using that the successor tuple $(p^\W_i,\rho^\W,\nu^\W,C^\W)$ of the $(p_i,v_i,w_i)$-witness $\W_i$ satisfies Requirement \textbf(c) and $(t_i,u_i)\in R(p_i)$,
            for each $i\in I$, there exists a $u'_i$ such that $u_i \leadsto_{w_i} u'_i$ and $(t'_i,u'_i)\in R(p^\W_i,\rho^\W,\nu^\W,C^\W)$.
          Thus by definition for each $i\in I$ there exists a standard open transition $\mathring{t}'_i$, two closed $\itt$-substitutions $\Trhoi',\Tnui'$ and $C'_i:\Var\to\Pow(\Tr\times\Tr)$, such that
            $\mathring{t}'_i[\Trhoi']=t'_i$, $\mathring{t}'_i[\Tnui']=u'_i$, $C'_i$ is consistent with $C^\W$ on $\var(p^\W_i)$,
            $\Trhoi',\Tnui'$ are consistent with $\rho^\W,\nu^\W$ respectively on $\var(p^\W_i)$ and are $C'_i$-ep-bisimilar on $\mathring{t}'_i$.
            Using \lem{renaming open transition}, we may assume that the transitions
            $\mathring{t}'_i$ for $i \in I$ use disjoint sets of transition variables and target variables, and
            employ no transition variables from $\Tvar(\mathring{t}) \cup \Tvar(\mathring{v})$
            or target variables from $\var(\mathring{t}) \cup \var(\mathring{v})$.

              When a substitution is applied to an object, only its restriction to the set of
              freely occurring variables matters. This allows us to unify substitutions.
              For example, since $\mathring{t},\mathring{v}$ do not have any target variables in common,
              $\var(\mathring{t}) \cap \var(\mathring{v})$ is equal to $\var(p)$,
              on which $\Trho^t$ and $\Trho^v$ are consistent with $\rho$, and thus with each other.
              Similarly, $\Tnu^t$ is consistent with $\Tnu^v$ on $\var(p)$.
              Neither do $\mathring{t},\mathring{v}$ share any transition variables.
              Thus we can construct a $\itt$-substitution $\Trho^*$ that is consistent with
              $\Trho^t$ on $\var(\mathring{t}),\Tvar(\mathring{t})$, and with $\Trho^v$ on
              $\var(\mathring{v}),\Tvar(\mathring{v})$,
                such that $\mathring{t}[\Trho^*]=\mathring{t}[\Trho^t]$ and $\mathring{v}[\Trho^*]=\mathring{v}[\Trho^v]$.
              In the same way, we can unify $\Tnu^t$ and $\Tnu^v$ into $\Tnu^*$.
              This also applies to $C^t,C^v$ -- we can construct $C^*$ consistent with $C$ on
              $\var(p)$ such that $\Trho^*,\Tnu^*$ are $C^*$-ep-bisimilar on $\mathring{t},\mathring{v}$. 
           
              In this spirit, we will create a $\itt$-substitution $\Trho'$ that is consistent with
              $\Trho^t$ on $\var(\mathring{t}),\Tvar(\mathring{t})$, with $\Trho^v$ on 
              $\var(\mathring{v}),\Tvar(\mathring{v})$ and, for each $i\in I$, with $\Trhoi'$ on
              $\var(\mathring{t}'_i),\Tvar(\mathring{t}'_i)$. Moreover, $\Trho'$ will be consistent
              with $\rho^\W$ on $\var(p^\W)$. In the same way, we unify $\Tnu^t$, $\Tnu^v$ and the
              $\Tnui'$ for $i \in I$ into a $\itt$-substitution $\Tnu'$ that is consistent
              with $\nu^\W$ on $\var(p^\W)$, and we unify  $C^t,C^v$ and the $C'_i$ into a function
              $C':\Var\to\Pow(\Tr\times\Tr)$ that is consistent with $C^\W$ on $\var(p^\W)$.
              To justify this we make the following observations.
              \begin{itemize}
                \item For each $i \in I$, $\Trhoi$ is
                        consistent with $\rho^\W$ on $\var(p^\W_i)$, $\Tnui$ with $\nu^\W$ and $C'_i$ with $C^\W$.
                \item For each $i,j\in I$, $\var(\mathring{t}'_i) \cap \var(\mathring{t}'_j)$ is equal to $\var(p^\W_i) \cap \var(p^\W_j)$,
                      on which $\Trhoi,\Trhoj$, $\Tnui,\Tnuj$ and $C'^i,C'_j$ are pairwise consistent.
                      The latter follows from the above.
                 \item Since $\rho^\W$ was defined as $\Trho^v{\upharpoonright}\Var$, $\Trho^v$ is
                      automatically consistent with $\rho^\W$ on $\var(p^\W)$. 
                      Likewise $\Tnu^v$ is consistent with $\nu^\W$ on $\var(p^\W)$, and $C^v$ with $C^\W$.
                \item For each $i\in I$, $\var(\mathring{v}) \cap \var(\mathring{t}'_i)$ is equal to $\var(p^\W_i)$, a subset of $\var(p^\W)$,
                  on which $\Trho^v,\Trhoi$, $\Tnu^v,\Tnui$ and $C^v,C'_i$ are pairwise consistent.
                      Again, the latter follows from the above.
                \item $\var(\mathring{t}) \cap \var(\mathring{v})$ is equal to $\var(p)$, on which
                  $\Trho^t,\Trho^v$, $\Tnu^t,\Tnu^v$ and $C^{\,t},C^v$ are pairwise consistent.
                \item For each $i\in I$, $\var(\mathring{t}) \cap \var(\mathring{t}'_i)$ is a subset of
                  $\var(p)\cap\var(p^\W_i)$, on which $\Trho^v,\Trhoi$, $\Tnu^v,\Tnui$ and $C^v,C'_i$ are pairwise consistent;
                 by transitivity we have that $\Trho^t,\Trhoi$, $\Tnu^t,\Tnui$ and
                 $C^{\,t},C'_i$ are pairwise consistent on this set.
                \item Among all these standard open transitions, no transition variable is shared.
              \end{itemize}
                Moreover, since $p^\W$ was defined as $\Otar(\mathring{v})$ and
                $\rho^\W$ was defined as $\Trho^v{\upharpoonright}\Var$,
                using $\var(\Otar(\mathring{v})) \subseteq \var(\mathring{v})$, any substitution
                consistent with $\Trho^v$ on $\var(\mathring{v})$ -- in particular $\Trho'$ --
                is automatically consistent with $\rho^\W$ on $\var(p^\W)$.
                Such an observation can also be made about $\Tnu'$ and $C'$. 
         
            Thus, the announced unifications can be made indeed.
            In particular, $C'$ is consistent with $C^\W$ on $\var(p^\W)$ and $\Trho',\Tnu'$ are consistent
            with $\rho^\W,\nu^\W$ respectively on $\var(p^\W)$. It also follows that
            $\Trho',\Tnu'$ are $C'$-ep-bisimilar on $\mathring{t},\mathring{v}$
            and on all $\mathring{t}'_i$ for $i\in I$.
            So when using this choice of $\rho^\W,\nu^\W$ and $C'$, 
            it remains to find a $\mathring{t}'\in\Oen(p^\W)$ such that $\mathring{t}'[\Trho']=t'$,
            $u \leadsto_w \mathring{t}'[\Tnu']$ and $\Trho',\Tnu'$ are $C'$-ep-bisimilar on $\mathring{t}'$.

  For simplicity, let us first deal with the case that $\mathring{r}$ contains no other
            variables than $x_i$ for $i=1,\dots,n$, $\tx_i$ for $i\in I_\snr$, $\ty_i$ for $i\in I_\sns$,
            $\tz_i$ for $i\in I$, and  $y'_i$ for $i\in I_\sns$.
            Let $\mathring{t}'$ be obtained from $\mathring{r}$ by substituting (for all appropriate $i$)
            $p_i$ for $x_i$,
            $\mathring{t}_i$ for $\tx_i$,
            $\mathring{v}_i$ for $\ty_i$,
            $\mathring{t}'_i$ for $\tz_i$, and
            $\Otar(\mathring{v}_i)$ for $y'_i$.
            Then $\mathring{t}'\in\Oen(p^\W)$, $\mathring{t}'[\Trho']=t'$, $\Trho',\Tnu'$ are $C'$-ep-bisimilar on $\mathring{t}'$, and we obtain
            \begin{equation}\label{intermediate}
              \frac{\{\mathring{t}_i  \leadsto_{\mathring{v}_i} \mathring{t}'_i \mid i\in I\}}
                   {\mathring{t} \leadsto_{\mathring{v}} \mathring{t}'}
            \end{equation}
            as a substitution instance of (\ref{rootRule}).
            Moreover (\ref{instance}) is obtained by applying the substitution $\Trho'$ to (\ref{intermediate}),
            while applying $\Tnu'$ to (\ref{intermediate}) yields
          \begin{equation}\label{instanceRight}
              \frac{\{u_i \leadsto_{w_i} u'_i \mid i\in I\}}{u \leadsto_v u'} 
          \end{equation}
            with $u' :=  \mathring{t}'[\Tnu']$. Since we had already that $u_i \leadsto_{w_i} u'_i$ 
            for $i\in I$, it follows that $u \leadsto_v u'$.

 	  In the more general case, $\mathring{r}$ may additionally contain transition variables
          $\tz_i \dblcolon x_i \goto{\mathit{za}_i} z'_i$ with $i \in \{1,...,n\}{\setminus}(I_\snr \cup I_\sns)$
          and $\tz_i \dblcolon y'_i \goto{\mathit{za}_i} z'_i$ with $i \in I_\sns{\setminus} I$.
          The transition $t'$ results from substituting transitions $\chi_i$ for these variables $\tx_i$,
          in addition to the already described substitution (for all appropriate $i$) of
            $p_i[\Trho']$ for $x_i$,
            $\mathring{t}_i[\Trho']$ for $\tx_i$,
            $\mathring{v}_i[\Trho']$ for $\ty_i$,
            $\mathring{t}'_i[\Trho']$ for $\tz_i$, and
            $\Otar(\mathring{v}_i)[\Trho']$ for $y'_i$.
          We have that $\chi_i \in \en(p^\W_i[\Trho']) = \en(p^\W_i[\rho^\W])$ for all $\tz_i$ occurring in $\mathring{r}$.
          Here $p^\W_i := \Otar(\mathring{v}_i)$ for $i\in I_\sns$ and $p^\W_i := p_i$ otherwise.
          For any transition $\chi_i$ enabled at $p^\W_i[\rho^\W]$, following the result of
          Part 1, and using that $\rho^\W,\nu^\W$ are two closed $\Sigma$-substitutions which are
          $C^\W$ bisimilar on $p^\W$, we can find a standard open transition $\mathring{\chi}_i\in\Oen(p^\W_i)$, two closed
          $\itt$-substitutions ${\Trho^\chi}_i,{\Tnu^\chi}_i$ and $C^\chi_i:\Var\to\Pow(\Tr\times\Tr)$ such
          that $\mathring{\chi}_i[{\Trho^\chi}_i]=\chi_i$, $C^\chi_i$ is consistent with $C^\W$ on $\var(p^\W_i)$, 
          ${\Trho^\chi}_i,{\Tnu^\chi}_i$ are consistent with $\rho^\W,\nu^\W$ respectively on $\var(p^\W_i)$
          and are $C^\chi_i$-ep-bisimilar on $\mathring{\chi}_i$.

            Using \lem{renaming open transition}, we may assume that the transitions
            $\mathring{\chi}_i$ use disjoint sets of transition variables and target variables, and
            employ no transition or target variables that are used in the application of $\Trho'$
            described earlier. The only variables these transition have in common, or share
            with $\mathring{t}$, $\mathring{v}$ or the $\mathring{t}_i$, are those in $\var(p^\W)$.
            Consequently, all ${\Trho^\chi}_i$ can be unified with $\Trho'$ into $\Trho''$,
            and in the same vain we obtain $\Trho''$ and $C''$.

            Let $\mathring{t}'$ be obtained from $\mathring{r}$ by substituting (for all appropriate $i$)
            $p_i$ for $x_i$,
            $\mathring{t}_i$ for $\tx_i$,
            $\mathring{v}_i$ for $\ty_i$,
            $\mathring{t}'_i$ for $\tz_i$ with $i \in I$,
            $\Otar(\mathring{v}_i)$ for $y'_i$, and now also
            $\mathring{\chi}_i$ for $\tz_i$ with $i \notin I$.
            Then $\mathring{t}'\in\Oen(p^\W)$, $\mathring{t}'[\Trho'']=t'$,
            $C''$ is consistent with $C^\W$ on $\var(p^\W)$,
            $\Trho'',\Tnu''$ are consistent with $\rho^\W,\nu^\W$ respectively on $\var(p^\W)$ and are $C''$-ep-bisimilar on
            $\mathring{t}'$, and once more we obtain (\ref{intermediate}) as a substitution instance of (\ref{rootRule}).
            Moreover (\ref{instance}) is obtained by applying the substitution $\Trho''$ to (\ref{intermediate}),
            while applying $\Tnu''$ to (\ref{intermediate}) yields (\ref{instanceRight})
            with $u' :=  \mathring{t}'[\Tnu'']$. Since we had already that $u_i \leadsto_{w_i} u'_i$ 
            for $i\in I$, it again follows that $u \leadsto_v u'$.
  
    \item The case where $p=\rec{X|S}$ is handled in the same way as above.\qed
  \end{itemize}
\end{proofThm}

\section{Recursive Definition Principle}
The \emph{Recursive Definition Principle (RDP)} \cite{BW90} says that each recursive specification has a solution.
  This is a desirable property of process algebras equipped with a semantic equivalence. In settings with a recursion
  construct $\rec{X|S}$, RDP is often formulated as the statement that the recursive $\rec{X|S}$ itself is such a solution.
  Here we show that RDP holds for enabling preserving bisimilarity in any De Simone language. 

\begin{observation}\label{obs:target of indicator transition}\upshape
  Fix a TSS $\tss$ in De Simone format.
  If $\ell(t)\in\Lab\setminus\Act$ then $\target(t)=\source(t)$.
\end{observation}

\begin{observation}\label{obs:successor of indicator transition}\upshape
  Fix a TSSS $\tsss$ in De Simone format.
  If $\ell(u)\in\Lab\setminus\Act$ and $t \leadsto_u v$ then $v=t$.
\end{observation}

\begin{theoremc}{RDP}\label{thm:rdp}\upshape
  If $S$ is a recursive specification such that $X\in V_S$ and $\var(S_Y)\subseteq V_S$ for all $Y\in V_S$,
    then $\rec{X|S} \bisep \rec{S_X|S}$.
\end{theoremc}
\begin{proof}
  Let $\R\subseteq S\times S\times\Pow(\Tr\times\Tr)$ be the smallest relation satisfying
    \begin{itemize}
      \item if $(p,q,R)\in\R'$ for some ep-bisimulation $\R'$ then $(p,q,R)\in\R$, and
      \item if $S$ is a recursive specification such that $X\in V_S$ and $\var(S_Y)\subseteq V_S$ for all $Y\in V_S$, then\linebreak
              $(\rec{X|S},\rec{S_X|S},R_{S,X})\in\R$ where
              \[\begin{aligned}
                R_{S,X} := & \{(\RecAct(X,S,t),t)\mid t\in\en(\rec{S_X|S}) \land \ell(t)\in\Act\} \cup {} \\
                           & \{(\RecIn(X,S,u),u)\mid u\in\en(\rec{S_X|S}) \land \ell(u)\in\Lab\setminus\Act\}.
              \end{aligned}\]
    \end{itemize}
  It suffices to show that $\R$ is an ep-bisimulation.
  That means, each tuple in $\R$ satisfies the requirements in \df{ep-bisimilarity}.

  If the tuple stems from the first clause then it must have satisfied all the requirements.

  If it stems from the second clause, it has the form $(\rec{X|S},\rec{S_X|S},R_{S,X})$ for certain $S,X$.
  
  \vspace{1.5ex}
  \noindent
  \textbf{Part 1.}
  By \df{transition expression},
    all transitions enabled at $\rec{X|S}$ have the form $\RecAct(X,S,t)$ or $\RecIn(X,S,u)$
    for some $t,u\in\en(\rec{S_X|S})$ with $\ell(t)\in\Act$ and $\ell(u)\in\Lab\setminus\Act$, such that
    \begin{itemize}
      \item $\RecAct(X,S,t)\in\en(\rec{X|S})$ iff $t\in\en(\rec{S_X|S}) \land \ell(t)\in\Act$, and
      \item $\RecIn(X,S,u)\in\en(\rec{X|S})$ iff $u\in\en(\rec{S_X|S}) \land \ell(u)\in\Lab\setminus\Act$.
    \end{itemize}
  Assuming the existence of corresponding transitions, \df{TSS in De Simone format} further gives us $\ell(\RecAct(X,S,t))=\ell(t)$ and $\ell(\RecAct(X,S,u))=\ell(u)$.
  Thus Part 1 of \df{ep-bisimilarity} holds.
  
  \vspace{1.5ex}
  \noindent
  \textbf{Part 2.}
  Suppose $(v,w)\in R_{S,X}$.
  If $\ell(v)=\ell(w)\in\Act$ then $v=\RecAct(X,S,w)$.
  From \dfs{\ref{df:transition expression} and \ref{df:TSS in De Simone format}} we have $\target(v)=\target(\RecAct(X,S,w))=\target(w)$.
  Therefore, $(\target(v),\target(w),{\it Id}_{\target(v)})\in\R$ with
  ${\it Id}_{\target(v)}=\{(t',t') \mid t'\in\en(\target(v))\}$ by the first
  clause in the definition of $\R$. (See $\R_{\it id}$ in the reflexivity proof
    of~\cite[Proposition~10]{GHW21ea}.)
  
  If $(t,u)\in R_{S,X}$,
    then $t=\mathit{rec}_\chi(X,S,u)$ where $\mathit{rec}_\chi{=}\RecAct$ if $\ell(t){=}\ell(u)\in\Act$ and $\mathit{rec}_\chi{=}\RecIn$ otherwise.
  \begin{itemize}
    \item Suppose $t \leadsto_v t'$ for some $t'$.
          By \df{TSSS in De Simone format}, the root of the proof $\pi$ of the successor literal $t \leadsto_v t'$ must apply a substitution instance
            \[
              \frac{u \leadsto_w t'}{t \leadsto_v t'}
            \]
            of the successor rule
            \[
              \frac{(\tx \dblcolon S_X \goto{\ell(u)} x') \leadsto_{(\ty \dblcolon S_X \goto{\ell(w)} y')} (\tz \dblcolon y' \goto{\ell(t')} z')}{\mathit{rec}_\chi(X,S,\tx \dblcolon S_X \goto{\ell(u)} x') \leadsto_{\RecAct(X,S,\ty \dblcolon S_X \goto{\ell(w)} y')} (\tz \dblcolon y' \goto{\ell(t')} z')}.
            \]
          in (our enrichment of the) De Simone format.
          By deleting the root node of $\pi$, we obtain a (sub)proof of the successor literal $u \leadsto_w t'$.
          Thus $u \leadsto_w t'$ and $(t',t')\in {\it Id}_{\target(v)}$.
    \item Suppose $u \leadsto_w u'$ for some $u'$.
          By applying the substitution instance 
            \[
              \frac{u \leadsto_w u'}{t \leadsto_v u'}
            \]
          of the above successor rule one obtains $t \leadsto_v u'$. Again, $(u',u')\in {\it Id}_{\target(v)}$.
  \end{itemize}
  If $\ell(v)=\ell(w)\in\Lab\setminus\Act$ then $v=\RecIn(X,S,w)$.
  By \obs{target of indicator transition}, $\target(v)=\source(v)=\rec{X|S}$; $\target(w)=\source(w)=\rec{S_X|S}$.
  Therefore, $(\target(v),\target(w),R_{S,X})\in\R$ by the second clause in the definition of $\R$.
  
  If $(t,u)\in R_{S,X}$,
    then $t=\mathit{rec}_\chi(X,S,u)$ where $\mathit{rec}_\chi{=}\RecAct$ if $\ell(t){=}\ell(u)\in\Act$ and $\mathit{rec}_\chi{=}\RecIn$ otherwise.
  \begin{itemize}
    \item Suppose $t \leadsto_v t'$ for some $t'$.
          By \df{TSSS in De Simone format}, the root of the proof $\pi$ of the successor literal $t \leadsto_v t'$ must apply a substitution instance
            \[
              \frac{u \leadsto_w u'}{t \leadsto_v t}
            \]
            of the successor rule
            \[
              \frac{(\tx \dblcolon S_X \goto{\ell(u)} x') \leadsto_{(\ty \dblcolon S_X \goto{\ell(w)} y')} (\tz \dblcolon y' \goto{\ell(u')} z')}{\mathit{rec}_\chi(X,S,\tx \dblcolon S_X \goto{\ell(u)} x') \leadsto_{\RecIn(X,S,\ty \dblcolon S_X \goto{\ell(w)} y')} \mathit{rec}_\chi(X,S,\tx \dblcolon S_X \goto{\ell(u)} x')}
            \]
          in (our enrichment of the) De Simone format.
          By deleting the root node of $\pi$, we obtain a (sub)proof of the successor literal $u \leadsto_w u'$.
          Thus $u \leadsto_w u'$.
          Moreover, from \obs{successor of indicator transition} we have $t'=t$ and $u'=u$.
          And $(t,u)\in R_{S,X}$ by assumption.
        \item Suppose $u \leadsto_w u'$ for some $u'$.
          Then $u'=u$ by \obs{successor of indicator transition}.
          By applying the substitution instance 
            \[
              \frac{u \leadsto_w u}{t \leadsto_v t}
            \]
          of the above successor rule one obtains $t \leadsto_v t$.   And $(t,u)\in R_{S,X}$ by assumption.
\qed
  \end{itemize}
\end{proof}

\end{document}